\documentclass[acmsmall,screen,nonacm]{acmart}
\usepackage[utf8]{inputenc}
\usepackage[T1]{fontenc}
\usepackage[english]{babel}
\usepackage{csquotes}
\usepackage{hyperref}
\usepackage{natbib}
\usepackage{wasysym}
\usepackage{placeins}
\usepackage[capitalize,noabbrev]{cleveref}
\usepackage{trimclip}
\usepackage{wrapfig}
\usepackage{amsmath}
\usepackage{amsthm}
\usepackage{mathpartir}
\usepackage{mathtools}
\usepackage{stmaryrd}
\usepackage{tcolorbox}
\usepackage{bbm}
\usepackage{bbold}
\usepackage{xspace}
\usepackage{xstring}
\usepackage{suffix}
\usepackage[shortlabels]{enumitem}
\usepackage{tikz}
\usetikzlibrary{arrows.meta,calc,positioning,tikzmark,shapes.multipart,cd,decorations.pathreplacing,decorations.pathmorphing}
\usepackage{tikz-cd}
\usepackage{mdframed}
\usepackage{aliascnt} 

\makeatletter
\def\@parfont{\itshape}
\makeatother

\crefname{section}{\S}{\S}

\newif\ifcomments
\commentsfalse

\newif\ifappendix
\appendixtrue

\colorlet{LightGray}{gray!40!}

\newcommand{\gbox}[1]{
    \tcbox[on line,size=fbox,
           colframe=white,colback=LightGray]{\ensuremath{#1}}}

\ifcomments
\paperwidth    =\dimexpr\paperwidth     + 4cm\relax
\oddsidemargin =\dimexpr\oddsidemargin  + 2cm\relax
\evensidemargin=\dimexpr\evensidemargin + 2cm\relax
\marginparwidth=\dimexpr\marginparwidth + 2.2cm\relax
\usepackage{todonotes}
\newcommand{\todobox}[3]{{\todo[color=#1!20,bordercolor=#1!80]{\footnotesize\textbf{#2:} #3}}}
\newcommand{\orpheas}[1]{\todobox{blue}{Orpheas}{#1}}
\newcommand{\sam}[1]{\todobox{red}{Sam}{#1}}
\newcommand{\ohad}[1]{\todobox{magenta}{Ohad}{#1}}
\newcommand{\cristina}[1]{\todobox{green}{Cristina}{#1}}
\newcommand{\paul}[1]{\todobox{yellow}{Paul}{#1}}
\else
\newcommand{\orpheas}[1]{}
\newcommand{\sam}[1]{}
\newcommand{\ohad}[1]{}
\newcommand{\cristina}[1]{}
\newcommand{\paul}[1]{}
\fi

\def\eg{e.g.,\xspace}
\def\ie{i.e.,\xspace}

\def\spac{\,}

\def\dash{-}

\def\tinym{\scriptscriptstyle}

\newcommand{\ndsize}{7}

\newcommand{\shpscale}{0.04*\ndsize}
\def\gap{0.1*\ndsize}
\def\vshift{0.5*\ndsize}
\def\hshift{0.8*\ndsize}

\tikzstyle{ptr}=[draw, thin, {Circle[scale=\shpscale]}-{Triangle[scale=\shpscale]}]
\tikzstyle{cptr}=[draw, thin, -{Triangle[scale=\shpscale]}]
\tikzstyle{ccell}=[draw,
    minimum width=\ndsize pt, minimum height=0.5*\ndsize pt, outer sep=0pt, inner sep=0pt,
    label={[shift={(-0.75*\ndsize pt,-0.35*\ndsize pt)},color=black,scale=\shpscale]#1}]
\tikzstyle{vcell}=[draw, dashed,
    minimum width=\ndsize pt, minimum height=0.5*\ndsize pt, outer sep=0pt, inner sep=0pt,
    label={[shift={(-0.75*\ndsize pt,-0.4*\ndsize pt)},color=black,scale=\shpscale]#1}]

\tikzstyle{brace}=[draw,decorate,decoration={brace,amplitude=5pt,raise=5ex}]
\tikzstyle{labbrace}=[midway,xshift=3.5em,yshift=1pt]

\newcommand{\pbsquare}[1]{
  \IfSubStr{#1}{l}
    {\IfSubStr{#1}{u}{\arrow[#1, phantom, pos=0.15, "{\ulcorner}"]}{\arrow[#1, phantom, pos=0.15, "{\llcorner}"]}}
    {\IfSubStr{#1}{u}{\arrow[#1, phantom, pos=0.15, "{\urcorner}"]}{\arrow[#1, phantom, pos=0.15, "{\lrcorner}"]}}
}
\newcommand{\posquare}[1]{
  \IfSubStr{#1}{u}
    {\arrow[#1, phantom, pos=0.15, "{\ulcorner}"]}
    {\arrow[#1, phantom, pos=0.15, "{\llcorner}"]}
}

\tikzset{rmonotip/.tip={Glyph[glyph math command=rto]}}
\tikzset{alloctip/.tip={Glyph[glyph math command=ato]}}
\tikzset{malloctip/.tip={Glyph[glyph math command=mato]}}

\newenvironment{arrow-diagram}[1][normal]{
\begin{tikzcd}[
    cramped,
    nodes={every node/.style={inner sep=0pt, outer sep=0pt}},
    arrows={start anchor=real center, end anchor=real center, shorten >= 1pt, shorten <= 1pt},
    sep={#1}]
}{\end{tikzcd}}

\def\north{(0,\vshift pt)}
\def\nnorth{(0,\vshift+\gap pt)}

\def\west{(-\hshift pt,0)}
\def\wwest{(-0.4*\hshift-\hshift pt,0)}
\def\east{(\hshift pt,0)}
\def\eeast{(0.4*\hshift+\hshift pt,0)}
\def\south{(0,-\vshift pt)}
\def\ssouth{(0,-\vshift+\gap pt)}

\newenvironment{proofsketch}{\noindent\textit{Proof Sketch.}}{}

\def\eqdef{:=}
\def\defeq{=:}
\def\inddef{::=}

\def\iso{\cong}

\newcommand{\img}[1]{\mathsf{Img}(#1)}
\newcommand{\dom}[1]{\mathsf{Dom}(#1)}

\WithSuffix\newcommand\inv^[1]{#1^{\text{-}1}}

\WithSuffix\newcommand\dag^[1]{{#1}^\dagger}

\newcommand{\xfrom}[1]{\xleftarrow{#1}}

\def\finpto{\rightharpoonup_{\mathsf{fin}}}
\newcommand\xto\xrightarrow
\def\injto{\rightarrowtail}
\def\inclto{\hookrightarrow}
\newcommand{\xinclto}[1]{\xhookrightarrow{\smash{#1}}}
\def\surjto{\twoheadrightarrow}

\newcommand{\xrto}[1]{\xto{#1}_{\mathrlap{\hspace{-1.5pt}{\scriptscriptstyle\mathsf{r}}}}}
\def\rto{\xrto\relax}
\newcommand{\xato}[1]{\xto{#1}_{\mathrlap{\hspace{-2.7pt}{\scriptscriptstyle\mathrm\nu}}}\,}
\def\ato{\xato\relax}

\def\To{\Rightarrow}
\def\xTo{\xRightarrow}

\def\dashto{\dashrightarrow}
\makeatletter
\newcommand{\xdashrightarrow}[2][]{\ext@arrow 0359\rightarrowfill@@{#1}{#2}}
\newcommand{\xdashleftarrow}[2][]{\ext@arrow 3095\leftarrowfill@@{#1}{#2}}
\newcommand{\xdashleftrightarrow}[2][]{\ext@arrow 3359\leftrightarrowfill@@{#1}{#2}}
\def\rightarrowfill@@{\arrowfill@@\relax\relbar\rightarrow}
\def\leftarrowfill@@{\arrowfill@@\leftarrow\relbar\relax}
\def\leftrightarrowfill@@{\arrowfill@@\leftarrow\relbar\rightarrow}
\def\arrowfill@@#1#2#3#4{%
  $\m@th\thickmuskip0mu\medmuskip\thickmuskip\thinmuskip\thickmuskip
   \relax#4#1
   \xleaders\hbox{$#4#2$}\hfill
   #3$%
}
\makeatother

\makeatletter
\newcommand{\partsmash}[2][tb]{%
  \def\mb@t{\ht\z@ #2\ht\z@}\def\mb@b{\dp\z@ #2\dp\z@}%
  \def\mb@tb{\mb@t \mb@b}%
  \edef\finsm@sh{\csname mb@#1\endcsname\box\z@}%
  \ifmmode \@xp\mathpalette\@xp\mathsm@sh
  \else \@xp\makesm@sh
  \fi}
\def\xsrightarrowfill@{\arrowfill@\relbar\relbar{\partsmash[t]{0.6}\rightarrow}}
\newcommand{\xsrightarrow}[1]{\ext@arrow 0079\xsrightarrowfill@{}{#1}}
\def\xsleftarrowfill@{\arrowfill@{\partsmash[t]{0.6}\leftarrow}\relbar\relbar}
\newcommand{\xsleftarrow}[1]{\ext@arrow 0079\xsleftarrowfill@{}{\hspace{1pt}#1}}
\def\xshookrightarrowfill@{\arrowfill@\relbar\relbar{\partsmash[t]{0.6}\hookrightarrow}}
\newcommand{\xshookrightarrow}[1]{\ext@arrow 0079\xshookrightarrowfill@{}{#1}}

\newcommand{\xsto}[1]{\xsrightarrow{#1}}
\newcommand{\xsfrom}[1]{\xsleftarrow{#1}}
\newcommand{\xsato}[1]{\xsto{#1}_{\mathrlap{\hspace{-2.7pt}{\scriptscriptstyle\mathrm\nu}}}\,}
\makeatletter
\newcommand{\xrightarrowtail}[2][]{%
  \ext@arrow 0359\rightarrowtailfill@{#1}{#2}%
}
\newcommand{\rightarrowtailfill@}{%
    \arrowfill@{\hspace{2pt}\relbar}\relbar{\partsmash[t]{0.6}\rightarrowtail}
}
\makeatother
\newcommand\xinjto{\xrightarrowtail}

\def\mwhere{\mathit{where}\;}
\def\mif{\mathit{if}\;}
\def\motherwise{\mathit{otherwise}}
\def\miff{\mathit{iff}\;}

\newcommand{\tlistr}[2]{(#1)_{#2}}

\def\va{\mathit{a}}
\def\vb{\mathit{b}}
\def\vc{\mathit{c}}
\def\vd{\mathit{d}}
\def\ve{\mathit{e}}
\def\vf{\mathit{f}}

\def\vh{\mathit{h}}
\def\vi{\mathit{i}}
\def\vj{\mathit{j}}
\def\vk{\mathit{k}}
\def\vl{\mathit{l}}

\def\vn{\mathit{n}}

\def\vp{\mathit{p}}
\def\vq{\mathit{q}}
\def\vr{\mathit{r}}
\def\vs{\mathit{s}}

\def\vu{\mathit{u}}

\def\vw{\mathit{w}}
\def\vx{\mathit{x}}
\def\vy{\mathit{y}}
\def\vz{\mathit{z}}
\def\vA{\mathit{A}}
\def\vB{\mathit{B}}
\def\vC{\mathit{C}}
\def\vD{\mathit{D}}

\def\vF{\mathit{F}}

\def\vI{\mathit{I}}
\def\vJ{\mathit{J}}

\def\vM{\mathit{M}}
\def\vN{\mathit{N}}

\def\vR{\mathit{R}}
\def\vS{\mathit{S}}

\def\vU{\mathit{U}}
\def\vV{\mathit{V}}
\def\vW{\mathit{W}}
\def\vX{\mathit{X}}
\def\vY{\mathit{Y}}
\def\vZ{\mathit{Z}}

\def\me{\varepsilon}

\def\Succ{\mathit{S}\,}

\def\mvarH{\vh}
\def\mvarM{\eta}

\def\setDiff{\setminus}
\def\setInt{\cap}

\def\fin{\mathsf{fin}}

\newcommand{\card}[1]{|#1|}
\newcommand{\downset}[1]{[#1]}

\def\quot{q}
\newcommand\eqclass[1]{q_{#1}}
\newcommand\meduplus[1]{\scalebox{1}{$\displaystyle\biguplus_{#1}$}}
\newcommand\disunion[1]{\displaystyle\meduplus{#1}}

\def\go{\mathit{go}}

\newcommand{\supp}[1]{\underbar{$#1$}}

\def\lambdaref{\ensuremath{\lambda_{\typekw{ref}}}}

\newcommand{\exprkw}[1]{{\rm {\bf #1}}}
\newcommand{\weaken}[2]{\prescript{#1}{}{\hspace{-1pt}#2}}

\def\True{\exprkw{true}}
\def\False{\exprkw{false}}
\newcommand{\Lam}[1]{\lambda#1.\spac}

\def\App{\spac}

\def\Unit{()}
\newcommand{\VPair}[2]{( #1 , #2 )}
\newcommand{\Pair}{\VPair}

\newcommand{\Split}[3]{\exprkw{split}\spac{#1}\spac\exprkw{as}\spac\Pair{#2}{#3}.\spac}

\newcommand{\Let}[2]{\exprkw{let}\,\mathit{#1}={#2}\spac\exprkw{in}\spac}
\WithSuffix\newcommand\Inj_[1]{\exprkw{inj}_{#1}\spac}
\def\InjL{\exprkw{inl}\spac}
\def\InjR{\exprkw{inr}\spac}
\newcommand{\Case}[5]{\exprkw{case}\spac{#1}\spac\{ #2.\spac #3 ; #4. \space #5 \}}
\newcommand{\CaseV}[5]{
\begin{array}[t]{@{}l}
    \exprkw{case}\spac{#1}\spac\exprkw{of} \\
    \quad#2. \spac #3 \\
    \quad#4. \spac #5
\end{array}}
\newcommand{\Absurd}[1]{\exprkw{case}\spac{#1}\{\}}

\def\Write{\spac\exprkw{:=}\spac}
\def\Read{\exprkw{!}\spac}
\newcommand\LetrefS[3]{\exprkw{letref}\spac{#1}\Write\spac{#2}\spac\exprkw{in}\spac{#3}}
\newcommand\Letref[2]{\exprkw{letref}\spac{#1}\spac\exprkw{in}\spac{#2}}

\def\Eq{\spac\exprkw{=\!=}\spac}

\newcommand{\typekw}[1]{\textsf{#1}}

\def\tyZero{\typekw{0}}
\def\tyUnit{\typekw{1}}
\def\tyBool{\typekw{2}}

\newcommand{\tyRef}[1]{\typekw{ref}_{#1}}
\def\ctypekw{\mathit{ctype}}
\newcommand\ctype[1]{\ctypekw\spac#1}

\def\tySum{\spac\typekw{$+$}\spac}
\def\tyProd{\spac\typekw{$\times$}\spac}

\def\tyFun{\rightarrow}

\def\sortZero{\mathsf{zero}}
\def\sortUnit{\mathsf{unit}}
\def\sortBool{\mathsf{bool}}
\def\sortNat{\mathsf{nat}}
\def\sortNatTwo{\mathsf{nat2}}
\def\sortLBool{\mathsf{boolList}}

\def\ctxNil{\cdot}

\newcommand{\typedv}[3]{#1 \vdash^{\mathsf{v}}_{#2} \,: #3}
\newcommand{\vtyped}[4]{#1 \vdash^{\mathsf{v}}_{#2} #3 : #4}

\newcommand{\typed}[4]{#1 \vdash_{#2} #3 : #4}

\newcommand{\RULEDEF}[1]{\hypertarget{#1}{{(#1)}}\xspace}
\newcommand{\RULE}[1]{\hyperlink{#1}{(#1)}\xspace}

\def\nameTExt{WorldExt}
\def\nameTVar{Var}
\def\nameTLam{Lam}
\def\nameTUnit{Unit}
\def\nameTVPair{VPair}
\def\nameTVInj{VInj}
\def\nameTVal{Val}
\def\nameTApp{App}
\def\nameTEPair{EPair}
\def\nameTSplit{Split}
\def\nameTEInj{EInj}

\def\nameTCase{Case}
\def\nameTLoc{Loc}
\def\nameTRead{Read}
\def\nameTWrite{Write}
\def\nameTEq{Eq}
\def\nameTAlloc{Alloc}

\def\nameRVal{Red-Val}
\def\nameRApp{Red-App}
\def\nameREPair{Red-EPair}
\def\nameRSplit{Red-Split}
\def\nameREInj{Red-EInj}
\def\nameRCase{Red-Case}
\def\nameRRead{Red-Read}
\def\nameRWrite{Red-Write}
\def\nameREq{Red-Eq}
\def\nameRAlloc{Red-Alloc}

\def\mcatI{\mathcal{I}}
\def\mcatC{\mathbbm{C}}

\def\comp{\circ}

\newcommand{\objC}[1]{\mathsf{Obj}({#1})}
\def\setC{\mathsf{Set}}

\newcommand{\funC}[2]{[#1, #2]}

\newcommand{\famC}[1]{\mathsf{Fam}({#1})}

\newcommand{\opC}[1]{{#1}^{\mathsf{op}}}

\newcommand{\commaC}[2]{{#1}\downarrow{#2}}
\newcommand{\sliceC}[2]{{#1}/{#2}}
\newcommand{\arrC}[1]{\mathsf{Arr}({#1})}

\def\finC{\mathsf{Fin}}

\def\worldC{\mathbbm{W}}
\def\pworldC{\mathbf{W}}
\def\inst{\mathbbm{I}}
\def\pinst{\mathbf{I}}
\def\instC{\worldC\inst}
\def\pinstC{\pworldC\pinst}
\def\tempC{\mathbbm{T}}

\def\tinstC{{\tempC\inst}}
\newcommand\allocC[1]{\mathsf{alloc}(#1)}

\newcommand{\homF}[3]{{#1}({#2}, {#3})}

\newcommand{\Colimit}[1]{\mathit{colimit}\,{#1}}
\newcommand{\Prod}[1]{\prod_{#1}\spac}
\newcommand{\Coprod}[1]{\coprod_{#1}\spac}
\newcommand{\Exp}[2]{{#1}^{#2}}
\newcommand{\End}[1]{\displaystyle{\int_{#1}}\spac}
\newcommand{\Coend}[1]{\displaystyle{\int^{\ensuremath{#1}}}\spac}
\newcommand{\pltimes}[1]{\otimes_{#1}}

\newcommand{\liplus}[1]{\oplus_{#1}}

\newcommand\indF[1]{|{#1}|}
\def\monFLS{\mathcal{T}}
\def\monFLSFin{\monFLS_{\!\fin}}
\def\heapF{\mathbbm{H}}
\def\dheapF{\underline\heapF}
\def\theapF{{\tempC\heapF}}

\def\hideMon{{\mathbbm{P}}}

\newcommand\stencilF[1]{\mathcal{S}_{#1}}

\newcommand{\inclITI}[1]{1_{#1}}
\newcommand{\extF}[2]{(#1 \oplus #2)}
\newcommand{\Plus}[1]{{#1}^+}

\newcommand{\Minus}[1]{{#1}^-}

\def\instU{\vu}
\def\tinstU{\vu}
\newcommand{\concL}[1]{\lfloor#1\rfloor}
\newcommand{\varL}[1]{\lceil#1\rceil}
\newcommand{\anyL}[1]{[#1]}

\newcommand{\uncon}[1]{\mathsf{uc}(#1)}
\newcommand{\con}[1]{\mathsf{c}(#1)}

\newcommand{\ureach}{\mathsf{ur}}
\newcommand{\reach}{\mathsf{r}}
\newcommand{\GU}[2]{\mathsf{GU}(#1, #2)}

\def\bind{\ensuremath{>\!\!>\!\!=\!}}

\WithSuffix\newcommand\ext^[1]{#1^*}
\WithSuffix\newcommand\res_[1]{{#1}_*}
\WithSuffix\newcommand\res|[1]{|^{#1}}
\WithSuffix\newcommand\dres|[1]{|_{#1}}
\WithSuffix\newcommand\uext^[1]{#1\rangle}
\WithSuffix\newcommand\proj_[1]{\pi_{#1}}
\newcommand{\pair}[2]{\langle#1, #2\rangle}

\newcommand{\triple}[3]{\langle#1, #2, #3\rangle}
\newcommand{\copair}[2]{[#1, #2]}

\newcommand\inj[1]{\mathit{inj}{\raisebox{-5pt}{$#1$}}}
\WithSuffix\newcommand\inj_[1]{\mathit{inj}_{#1}}
\WithSuffix\newcommand\inj^[1]{\mathit{i}^{#1}}
\def\curry{\mathsf{curry}}
\def\uncurry{\mathsf{uncurry}}
\def\incl{\rotatebox[origin=c]{90}{$\hookrightarrow$}}

\def\eval{\mathsf{eval}}
\def\swap{\mathsf{swap}}
\WithSuffix\newcommand\incl_[1]{\mathit{i}_{\tinym#1}}
\def\str{\mathsf{str}}

\def\dstr{\mathsf{dstr}}
\def\distr{\mathsf{distr}}

\def\cov{\mathsf{cov}}
\def\rebase{\mathsf{rebase}}
\def\embed{\mathsf{embed}}
\def\ev{\mathsf{ev}}
\def\hideSt{\mathsf{hide}}
\newcommand{\compl}[1]{{#1}^\mathsf{c}}

\newcommand{\hideRes}[2]{\mathsf{hide}^{#1}_{#2}}
\def\meZ{\me_{0}}
\def\meS{\me_{\mathsf{S}}}

\def\typeLoc{\mathbbm{L}}
\def\typeZero{\mathbb{0}}
\def\typeUnit{\mathbb{1}}
\def\typeProd{\times}
\def\typeSum{+}
\def\typeBool{\mathbbm{B}}
\def\typeNat{\mathbbm{N}}
\def\typeVal{\mathsf{Val}}

\def\typeExpr{\mathsf{Expr}}
\def\typeSort{\mathsf{Sort}}
\def\typeFG{\mathsf{FGType}}
\def\typeType{\mathsf{Type}}

\def\typeRes{\mathsf{Res}}

\newcommand{\Subst}[2]{[#1 / #2]}
\newcommand{\SubstTwo}[4]{[#1 / #2, #3 / #4]}

\def\obsEq{\iso_{\mathsf{obs}}}

\newcommand{\sem}[1]{\llbracket#1\rrbracket}
\newcommand{\semV}[1]{\sem{#1}^{\mathsf{v}}}
\newcommand{\semM}[2]{\sem{#1}_{\!#2}}

\def\cycle{\mathsf{c}}
\newcommand{\cfg}[2]{\pair{#1}{#2}}
\newcommand{\bigstep}[4]{\pair{#1}{#2} \Downarrow \pair{#3}{#4}}

\def\eqOp{\mathit{eq}}
\def\genRead{\mathit{read}}
\def\genWrite{\mathit{write}}
\def\genAlloc{\mathit{alloc}}

\newcommand{\All}[1]{\forall #1.\spac}
\newcommand{\AllTwo}[2]{\forall #1,#2.\spac}

\theoremstyle{definition}
\newtheorem{theorem}{Theorem}
\theoremstyle{definition}
\newtheorem{definition}{Definition}[section]

\newaliascnt{lemma}{definition}
\newaliascnt{corollary}{definition}
\newaliascnt{example}{definition}

\theoremstyle{definition}
\newtheorem{lemma}[lemma]{Lemma}

\theoremstyle{definition}

\theoremstyle{remark}
\newtheorem{example}[example]{Example}

\newtheoremstyle{proof}
    {3pt}       
    {3pt}       
    {}          
    {}          
    {\itshape}  
    {.}         
    {.5em}      
    {}          

\aliascntresetthe{lemma}
\aliascntresetthe{corollary}
\aliascntresetthe{example}

\setcopyright{rightsretained}
\acmJournal{PACMPL}
\acmYear{2027} 
\acmVolume{11} 
\acmNumber{POPL} 
\acmArticle{0} 
\acmMonth{1}
\acmDOI{00.0000/0000000}

\begin{document}

\title{Finitary Semantics for Full Ground Local State}

\author{Orpheas van Rooij}
\orcid{0009-0009-7780-4772}
\affiliation{%
  \institution{University of Edinburgh}
  \country{UK}
}
\email{orpheas.vanrooij@ed.ac.uk}

\author{Ohad Kammar}
\orcid{0000-0002-2071-0929}
\affiliation{%
  \institution{University of Edinburgh}
  \country{UK}
}
\email{ohad.kammar@ed.ac.uk}

\author{Sam Lindley}
\orcid{0000-0002-1360-4714}
\affiliation{%
  \institution{University of Edinburgh}
  \country{UK}
}
\email{sam.lindley@ed.ac.uk}

\author{Cristina Matache}
\orcid{0009-0003-6036-6426}
\affiliation{%
  \institution{University of Birmingham}
  \country{UK}
}
\email{c.matache@bham.ac.uk}


\begin{abstract}
\emph{Full ground local state} (FGLS) refers to dynamically allocated mutable state that allows storing ground values and references.
It is a key ingredient in many imperative algorithms as it enables (cyclic) data structures.
In this work, we treat full ground local state as a computational effect, focusing on one particular denotational model: 
\citeauthor*{kammar2017monad}'s possible worlds monad on sets indexed over sets of locations.

We resolve an outstanding question regarding this FGLS monad: is it finitary?
We show that the FGLS monad is not finitary by showing the existence of non-finitary computations in the monad.
We then introduce a finitary submonad of~\citeauthor*{kammar2017monad}'s monad, give it a concrete description and show that it provides an adequate semantics for FGLS. 
The submonad we construct paves the way to understanding FGLS in the future via an equational axiomatization suitable for program reasoning.



\end{abstract}



\maketitle

\section{Introduction}

A ubiquitous feature of imperative programs is access to dynamically allocated mutable state:
reference-based memory that permits allocation of new memory cells, as well as reading and writing from existing ones. 
Dynamic allocation of mutable memory is often called \emph{local state} since allocated memory can only be accessed from those parts of the program that were passed the reference.
In the case of \emph{full ground} local state (FGLS), cells may store ground values (\eg booleans, integers) and references, 
enabling us to represent imperative data structures such as cyclic linked lists.
Crucially, full ground local state forbids storing functions which introduce general recursion through Landin's knot~\cite{landin1964mechanical}.

Various models for local state exist in the literature which include among others
relational models built via logical relations on nominal-cpos~\cite{benton2005relational,bohr2006relational}, 
step-indexed syntactic logical relations~\cite{ahmed2004semantics,birkedal2011step}, 
game semantic models~\cite{abramsky1998fully,laird2008game,murawski2012algorithmic},
as well as monad semantics based on possible worlds~\cite{moggi1990abstract,plotkin2002notions,kammar2017monad}.

We are concerned with possible worlds semantics that model heaps and types as sets indexed over sets of locations,
together with a suitable state-like monad on the appropriate category.
For the ground local state case where only ground values can be stored,
heaps are of a contravariant nature: we can select smaller segments of the heap by projecting out some of the cells \cite{moggi1990abstract,plotkin2002notions}. 
In the full ground setting this is no longer the case: projecting out a smaller heap may result in an invalid one since the references stored in the cells may now point to undefined data, \ie they may be dangling.
\Citet{kammar2017monad} address this by equipping heaps with a covariant structure: heaps can be extended with new data modelled via a category of instantiations that carries the data.
Their FGLS monad supports:
\begin{itemize}[leftmargin=10pt]
    \item[-] \emph{Effect masking:} Computations that do not depend or return references are semantically equivalent to pure values. 
    \item[-] \emph{Adequacy:} Denotationally equal terms are observationally indistinguishable.
                            \[ \sem\vM = \sem\vN \implies \vM \obsEq \vN \]
        where $\vM,\vN$ are $\lambda$-terms that manipulate full ground references.
    \item[-] \emph{Program Equivalences:} Validates desired equations such as
        \begin{align}\label{eq:ground-ls-eq}
              \LetrefS\vx\vV{\weaken\vx\vM} \;\equiv\; \vM
              \qquad & \qquad
              \vx \Write \vV; \vz \Write \weaken\vx\vW; \vM \;\equiv\; \vz \Write \vW; \vx \Write \weaken\vz\vV; \vM  
          \end{align}
          where $\LetrefS\vx\vN\vM$ allocates a new reference $\vx$ in $\vM$ with value $\vV$, 
          $\vX\Write\vV$ updates the contents of reference $\vx$ with $\vV$ and $\weaken\vx\vM$ states that $\vx$ does not occur free in $\vM$. 
\end{itemize}
These properties make it a useful model for FGLS and justify using it as basis for more sophisticated stores.
This paper extends the metatheory of this monad by resolving an outstanding question.

\paragraph{Finitary presentations of a monad}
In seminal work,~\citeauthor{moggi89computational} used monads to model computational effects such as mutable state, input/output and non-determinism.
The theory of algebraic effects~e.g.~\cite{plotkin2002notions,bauer2018what} refines Moggi's programme by using algebraic theories in the sense of universal algebra to characterise computational effects equationally.
In this setting, the effectful behaviour of programs is described by a set of effect operations and equations between these operations.
For example, non-determinism can be described by an operation $\mathsf{or}(x,y)$ that is commutative, associative and idempotent.


Moggi's monads can be recovered via the well-known correspondence between monads and algebraic theories~\cite{linton1966some,hyland2007category,robinson2002variations}: finitary algebraic theories correspond to finitary monads on the category of sets.
An algebraic theory is finitary when it admits a presentation whose operations are of finite arity.
There is no requirement that the set of operations or equations are finite.

However, when describing a computational effect, one often seeks a presentation of the algebraic theory with a small number of intuitive operations and equations, as in the non-determinism example above.
The role of the presentation is to serve as a basis for programming with the effect, via the operations, and for reasoning about effectul programs via the equations.

The goal of our line of research is to find such an algebraic presentation for the full ground local state monad of~\citeauthor{kammar2017monad}, for which none has been discovered yet.
Roughly speaking, the FGLS monad is defined on indexed sets, rather than sets, so we would need to use a more sophisticated notion of algebraic theory, rather than the classic one, to maintain the monad-theory correspondence.

Such extensions of algebraic theories and their correspondence to monads have been developed, for example: monads with arities~\cite{BERGER20122029}, enriched Lawvere theories~\cite{power1999enriched} and parameterized algebraic theories~\cite{staton2013instances}.
Out of these extensions, only parametrized algebraic theories admit a syntax for specifying presentations of theories, which we are particularly interested in from a programming perspective.
Moreover, both enriched Lawvere theories and parameterized algebraic theories require the monad in question to be at least finitary.

Until now, the question whether the full ground local state monad of~\citeauthor{kammar2017monad} is finitary remained open.
In this paper, we provide a negative answer.
We then characterize a finitary submonad of the FGLS monad and give it a concrete description.
Therefore, the submonad we introduce opens the way to characterizing FGLS as a parameterized algebraic theory, or as an enriched Lawvere theory, in future work. 
The main contributions of the paper are as follows:




\begin{enumerate}[(i)]
    \item \textit{Non-finitarity.}
          We show that the FGLS monad is non-finitary by showing the existence of certain non-finitary computations~(\cref{sec:overview-non-finitariness}).
          Such computations cannot be represented as terms in the programming language, 
          and so the model has extra information which allows denotations to detect observationally undetectable differences. 
    \item \textit{Worlds, Heaps with variables and Unification.}
          We develop the machinery needed to describe finitary FGLS computations.
          We introduce two-kinded worlds, heaps with variables, heap unification and matching and identify the relevant structures~(\cref{sec:theory-templates}).
    \item \textit{Finitary FGLS monad.}
          Using this machinery we define a finitary FGLS submonad and show that its functor and monad structure lifts from the full monad~(\cref{sec:fin-monad}). 
          Due to its finitary nature, this monad is more amenable to equational axiomatisations.
          It further enjoys adequacy of the denotational semantics with respect to the operational semantics, inherited from the non-finitary monad~(\cref{sec:fin-sem-fgls}).
\end{enumerate}


\section{Background and Overview}

\subsection[Programming Language]{Programming Language $\lambdaref$}\label{sec:pl-lamref}

We revisit the programming language $\lambdaref$, a higher-order $\lambda$-calculus with full ground references.
The formalism of $\lambdaref$ is parametrised by a set of sorts $\vC \in \typeSort$ which we use to define the set of full ground types as
\[\begin{array}{r@{\,}c@{\,}l}
    \typeFG \ni \vD & \inddef & \tyRef\vC \mid \tyZero \mid \vD \tySum \vD' \mid \tyUnit \mid \vD \tyProd \vD'
\end{array}\]
It includes a reference type $\tyRef\vC$ for every sort $\vC$,
the empty type $\tyZero$ and $\vA \tySum \vB$ representing finite sums,
and unit type $\tyUnit$ and $\vA \tyProd \vB$ representing finite products.

We further require an interpretation function $\ctypekw$ that assigns a full ground type $\vD$ for every sort $\vC$.
This stratification between sorts and full ground types enables us to represent potentially cyclic data structures where the recursive occurrence is guarded by a reference type.
\begin{example}
Consider $\typeSort \eqdef \{\sortBool, \sortLBool, \sortNat, \sortNatTwo  \}$.
The following interpretation for $\sortBool, \sortLBool$ represents boxed boolean valued linked lists:
\[ \ctype\sortBool \eqdef \tyUnit \tySum \tyUnit \qquad \ctype\sortLBool \eqdef \tyUnit \tySum \tyRef\sortBool \times \tyRef\sortLBool \]
and for $\sortNat$ represents boxed natural numbers and $\sortNatTwo$ are boxed pairs of boxed natural numbers:
\[ \ctype\sortNat \eqdef \tyUnit \tySum \tyRef\sortNat \qquad \ctype\sortNatTwo \eqdef \tyRef\sortNat \tyProd \tyRef\sortNat \]
\end{example}

\paragraph{Worlds}\label{sec:lamref-worlds}
Pointers to memory cells are called \emph{locations} ranging over a countable set $\typeLoc$,
and finite sets of locations with an associated sort are called worlds.
Formally, a world is a partial function $\vw : \typeLoc \finpto \typeSort$ with finite support. The finite support of $\vw$ is denoted by $\supp\vw$.
We write worlds as $\vw = \vl_1 : \vC_1, \hdots, \vl_\vn : \vC_\vn$, to really mean its support is $\supp\vw = \vl_1, \hdots, \vl_\vn$ and satisfies $\vw(\vl_\vi) = \vC_\vi$.
We write $\card\vw$ for the cardinality of $\vw$, which is $\vn$ in this case, and we order worlds by inclusion:
$\vw \leq \vw' \iff \All{\vl : \vC \in \vw} \vw'(\vl) = \vC$.
Lastly, by convention, when enumerating locations $\vl_1, \hdots, \vl_\vn$ we assume that they are all distinct.
%


\paragraph{Syntax}
$\lambdaref$ is a coarse-grained call-by-value language with the usual $\lambda$-calculus constructs for finite sums and products.
\Cref{fig:lambdaref-syntax} gives the syntax of $\lambdaref$ where the shaded constructs relate to references.
The types of $\lambdaref$ are the full ground types closed under the function type constructor, thus full ground types are a subset of all $\lambdaref$ types.
Values include variables and $\lambda$-abstractions, location literals $\vl$ which are pointers to allocated memory cells as above,
as well as the unit, binary pairs and injections.
Expressions consist of the usual constructs: values, application, pairs over expressions, pair eliminator $\Split\vM{\vx_1}{\vx_2}\vN$,
injections over expressions, falsity eliminator $\Absurd\vM$, and sum eliminator $\Case\vM{\vx_1}{\vN_1}{\vx_2}{\vN_2}$.
For references, $\vM\Eq\vN$ performs equality testing between two references, $\Read\vM$ reads a reference, $\vM\Write\vN$ writes value of $\vN$ to reference $\vM$,
and $\Letref{\vx_1 \Write \vV_1, \hdots, \vx_\vn \Write \vV_\vn}\vM$ allocates $\vn$ references with the corresponding values and makes them available to $\vM$.

We introduce the following sugar: $\tyBool = \tyUnit \tySum \tyUnit$ with $\True \eqdef \inj_1 \Unit, \False \eqdef \inj_2 \Unit$,
$\Let\vx\vM\vN \eqdef (\Lam\vx\vN)\vM$, and $\vM;\vN \eqdef \Let\vx\vM\vN$ where $\vx$ does not occur free in $\vN$.
\begin{figure}[ht]
\begin{align*}
    \typeType \ni \vA,\vB & \inddef \gbox{\tyRef\vC} \mid
                                    \tyZero \mid \vA \tySum \vB \mid
                                    \tyUnit \mid \vA \tyProd \vB \mid
                                    \vA \tyFun \vB                   & \tag{Reference and $\times,+,\to$ types}
    \\
    \typeVal \ni \vV,\vW  \inddef &\, \vx \mid \Lam\vx\vM \mid \gbox\vl \mid      & \tag{$\lambda$-calculus, and location literals} \\
                                  &\, \Unit \mid \Pair\vV\vW \mid                 & \tag{Products} \\
                                  &\, \Inj_{\vi \in \{1,2\}} \vV                  & \tag{Sums}
    \\
    \typeExpr \ni \vM,\vN  \inddef &\, \vV \mid \vM\App\vN \mid & \tag{Values and Application} \\
                                   &\, \Pair\vM\vN \mid \Split\vM{\vx_1}{\vx_2}\vN \mid & \tag{Products} \\
                                   &\, \Inj_{\vi \in \{1,2\}} \vM \mid \Absurd\vM \mid \Case\vM{\vx_1}{\vN_1}{\vx_2}{\vN_2} \mid & \tag{Sums} \\
                                   &\, \gbox{\vM\Eq\vN} \mid \gbox{\Read\vM} \mid \gbox{\vM\Write\vN} \mid \gbox{\Letref{\vx_1 \Write \vV_1, \hdots, \vx_\vn \Write \vV_\vn}\vM} & \tag{References}
\end{align*}
\Description{Syntax of programming language $\lambdaref$}
\caption{Syntax of $\lambdaref$}
\label{fig:lambdaref-syntax}
\end{figure}

Typing judgments $\vtyped\Gamma\vw\vV\vA$ for values and $\typed\Gamma\vw\vM\vA$ for expressions are parametrised by a world $\vw$ that represents the allocated locations.
This parameter is used in \RULE\nameTLoc to ensure that locations $\vl$ are present in the prescribed world.
The type system is monotone on world extensions $\vw \leq \vw'$, thus rule \RULE\nameTExt is admissible.
The interpretation function $\ctype$ is used in \RULE\nameTRead, \RULE\nameTWrite to obtain the full ground type associated with the sort.
Equality testing \RULE\nameTEq expects two expressions with reference types.
Allocation \RULE\nameTAlloc performs simultaneous recursive allocation, that is, makes the references to be allocated available in the instantiation data.
This allows the creation of mutually cyclic data structures.
\begin{figure}[ht]
\begin{mathpar}
\inferrule[\RULEDEF\nameTLoc]{\vl : \vC \in \vw}{\vtyped\Gamma\vw\vl{\tyRef\vC}}

\mprset{fraction={===}}
\inferrule[\RULEDEF\nameTExt]{
    \vw \leq \vw'
    \\\\
    \typed\Gamma\vw\vM\vA
}{\typed\Gamma{\vw'}\vM\vA}
\mprset{fraction={---}}

\inferrule[\RULEDEF\nameTRead]{\typed\Gamma\vw\vM{\tyRef\vC}}{\typed\Gamma\vw{\Read\vM}{\ctype\vC}}

\inferrule[\RULEDEF\nameTWrite]{
    \typed\Gamma\vw\vM{\tyRef\vC}
    \\\\
    \typed\Gamma\vw\vN{\ctype\vC}
}{\typed\Gamma\vw{\vM\Write\vN}\tyUnit}

\\
\inferrule[\RULEDEF\nameTEq]{
    \typed\Gamma\vw\vM{\tyRef\vC}
    \\\\
    \typed\Gamma\vw\vN{\tyRef\vC}
}{\typed\Gamma\vw{\vM\Eq\vN}\tyBool}

\inferrule[\RULEDEF\nameTAlloc]{
    (\typed{\Gamma, \vx_1 : \tyRef{\vC_1}, \hdots, \vx_\vn : \tyRef{\vC_\vn}}\vw{\vV_\vi}{\ctype\vC_\vi})_{i=1}^\vn
    \\\\
    \typed{\Gamma, \vx_1 : \tyRef{\vC_1}, \hdots, \vx_\vn : \tyRef{\vC_\vn}}\vw\vM\vA
}{\typed\Gamma\vw{\Letref{\vx_1 \Write \vV_1, \hdots, \vx_\vn \Write \vV_\vn}\vM}\vA}
\end{mathpar}
\Description{Typing rules of programming language $\lambdaref$ regarding references}
\caption{Typing rules of $\lambdaref$ regarding references}
\label{fig:typing-rules-lambdaref}
\end{figure}

\begin{example}
\renewcommand\ndsize{18}
    Creates a cyclic list with one node with the associated heap given on the right
\[\begin{array}{l@{\;\qquad}r}
    \Letref{(\vx : \tyRef\sortBool) \Write \True, (\vy : \tyRef\sortLBool) \Write \Inj_2\Pair\vx\vy}\vy
&
\begin{tikzpicture}[baseline={(lx.base)},remember picture]
    \node[ccell] (lx) {\tiny $\True$};
    \node[ccell] (ly) [right=15pt of lx] {\tiny $\pair{\bullet}{\bullet}$};
    \draw[cptr] ($ (ly.center) + (2pt,0pt) $) to ++\east to ++\north to ++\west to ++(0,-4pt);
    \draw[cptr] ($ (ly.center) + (-3pt,0pt) $) to ++\north to ++\wwest to ++\west to ++(0,-4pt);
\end{tikzpicture}
\end{array}\]
\end{example}

\paragraph{Operational Semantics}\label{sec:lambdaref-os}

An (untyped) heap $\mvarH$ is a partial function $\typeLoc \finpto \typeVal$, and a configuration $\cfg\vM\mvarM$ consists of an expression $\vM$ and a heap $\mvarM$.
The operational semantics of $\lambdaref$ is given in big-step, $\bigstep\vM\mvarM\vV{\mvarH'}$, which relates configurations on expressions with configurations on values.
\cref{fig:lambdaref-os} gives the rules related to references.
\begin{figure}[ht]
\begin{mathpar}
    \inferrule[\RULEDEF\nameRRead]{
        \bigstep\vM\mvarH\vl{\mvarH'}
        \\\\
        \mvarH'(\vl) = \vV
    }{
        \bigstep{\Read\vM}\mvarH\vV{\mvarH'}
    }

    \inferrule[\RULEDEF\nameRWrite]{
        \bigstep\vM\mvarH\vl{\mvarH'}
        \\\\
        \bigstep\vN{\mvarH'}\vW{\mvarH''}
    }{
        \bigstep{\vM \Write \vN}\mvarH\Unit{\mvarH''[\vl \mapsto \vW]}
    }

    \inferrule[\RULEDEF\nameREq]{
        \bigstep\vM\mvarH{\vl_1}{\mvarH'}
        \\\\
        \bigstep\vN{\mvarH'}{\vl_2}{\mvarH''}
    }{
        \bigstep{\vM \Eq \vN}\mvarH{\#(\vl_1 == \vl_2)}{\mvarH''}
    }

    \inferrule[\RULEDEF\nameRAlloc]{
        \vl_1, \hdots, \vl_\vn \not\in \dom\mvarH
        \\\\
        \theta = [\vx_\vi \mapsto \vl_\vi]_{\vi = 1}^\vn
        \and
        \mvarH' = \mvarH[\vl_\vi \mapsto \vV_\vi[\theta]]_{\vi = 1}^\vn
        \and
        \bigstep{\vM[\theta]}{\mvarH'}\vW{\mvarH''}
    }{
        \bigstep{\Letref{\vx_1 \Write \vV_1, \hdots, \vx_\vn \Write \vV_\vn}\vM}\mvarH\vW{\mvarH''}
    }
\end{mathpar}
\Description{Big-step style reduction rules of references for programming language $\lambdaref$}
\caption{Reduction rules of $\lambdaref$ for references}
\label{fig:lambdaref-os}
\end{figure}

\RULE\nameRRead evaluates $\vM$ to a location $\vl$ and returns its value provided the resulting heap $\mvarH'$ defines $\vl$.
\RULE\nameRWrite evaluates from left to right, reducing $\vM$ and $\vN$ to a location $\vl$ and new value $\vW$ and updates the heap $\vl$ with the new value.
\RULE\nameREq evaluates left to right, reducing $\vM$ and $\vN$ to locations $\vl_1$, $\vl_2$ and returns whether they are equal $\vl_1 == \vl_2$.
Function $\#(\vb)$ interprets boolean $\vb$ from the metalevel to syntax.
Rule \RULE\nameRAlloc generates distinct fresh locations $\vl_1,\hdots,\vl_\vn$, and evaluates the body $\vM$ with the updated heap after substituting the new locations.

Well typed terms in $\lambdaref$ are sound and total: a closed well-typed term evaluated with a compatible heap reduces to a closed value.
We refer the reader to \cite{kammar2017monad} for the full specification of $\lambdaref$ which we include in the appendix for completeness~\cref{sec:lambdaref}.

\subsection{Full Ground Local State Monad}\label{sec:overview-fgls}

We assume familiarity with categories, functors and natural transformations, and explain any further concepts as needed.
For a small category $\mcatC$, we denote the category of its set-valued functors $\mcatC \to \setC$ as $\funC\mcatC\setC$, and let $\homF\mcatC\vC\dash : \mcatC \to \setC$ be the covariant hom functor, \ie $\homF\mcatC\vC\vD$ is the set of all morphisms with domain $\vC$ and codomain $\vD$ in $\mcatC$.

The category $\worldC$ has the worlds $\vw$ (\cref{sec:lamref-worlds}) for objects,
and morphisms $\vw \xto\sigma \vw'$ are world extensions with renamings: $\sigma$ is an injective function between their supports $\supp\vw \injto \supp{\vw'}$ that preserve the sorts.
That is, for every location $\vl \in \supp\vw$, we have $\vw(\vl) = \vw'(\sigma(\vl))$.
Further, let $\vw' \ominus \sigma$ be the world of all locations in $\vw'$ that are not hit by $\sigma$: $\supp{\vw' \ominus \sigma} = \supp{\vw'} \setDiff \img\sigma$, and let $\vw' \ominus \sigma \xsto{\compl\sigma} \vw'$ be the inclusion.

\paragraph{Heaplets and Heaps}

Closed full ground values $\vtyped\ctxNil\vw\vV{\ctype\vC}$ are modelled semantically as elements $\sem\vV \in \sem{\ctype\vC}(\vw)$, 
where $\sem{\ctype\vC}$ is an indexed set over worlds \ie a functor in $\pworldC \eqdef \funC\worldC\setC$.
The monotonicity of the type system with respect to world extensions \RULE\nameTExt is modelled by the covariant structure of functors $\sem{\ctype\vC}$.
A world extension $\vw \leq \vw'$ is a morphism in $\worldC$, and thus the functorial action on morphisms takes a $\sem\vV : \sem{\ctype\vC}(\vw)$ and extends it to $\sem{\ctype\vC}(\vw \leq \vw')\sem\vV : \sem{\ctype\vC}(\vw')$.
Reference types are interpreted using the hom-functor $\sem{\tyRef\vC} = \homF\worldC{(\vl : \vC)}\dash$,
and thus $\sigma \in \sem{\tyRef\vC}(\vw)$ is a morphism $(\vl : \vC) \xto\sigma \vw$ in $\worldC$, \ie a selection of a single $\vC$-sorted location in $\vw$.

The \emph{heaplet} functor $\dheapF : \opC\worldC \times \worldC \to \setC$ represents heap segments with $\Minus\vw$ defined locations.
\[ \dheapF(\Minus\vw,\Plus\vw) = \Prod{\vl : \vC \in \Minus\vw} \sem{\ctype\vC}(\Plus\vw) \]
Values in a heaplet have access to a possibly bigger world $\Plus\vw$, enabling local level reasoning in a spirit similar to separation logic.
The heaplet functor is contravariant on world $\Minus\vw$ (we can project out a smaller heaplet), and covariant on world $\Plus\vw$ (by the covariant structure of semantic values).

\begin{example}\label{ex:overview-heaplets}
    Consider world extensions $(\vl_3 : \sortNatTwo) = \vw_1' \leq \vw_1 = (\vl_3 : \sortNatTwo, \vl_4 :\sortNat)$
    and $(\vl_1, \vl_4 : \sortNat, \vl_3 : \sortNatTwo) = \vw_2 \leq \vw_2' = (\vl_1, \vl_2, \vl_4 : \sortNat, \vl_3 : \sortNatTwo)$.
    Applying them to heaplet $\mvarM$ results in:
\begin{center}
\renewcommand\ndsize{15}
\begin{tikzpicture}[baseline={(l3.base)},remember picture]
    \node[vcell=$\vl_1$] (a1) [] {};
    \node[ccell=$\vl_3$] (l3) [below=0pt of a1] {\tiny $\pair\bullet\bullet$};
    \node[ccell=$\vl_4$] (l4) [below=0pt of l3] {\tiny $\InjL\Unit$};
    \node [below=10pt of l3] () {$\mvarM : \dheapF(\vw_1, \vw_2)$};
    \draw[cptr] ($ (l3.center) + (4pt,0pt) $) to ++\east to ++\north to ++(-9pt,0pt);
    \draw[cptr] ($ (l3.center) + (-4pt,0pt) $) to ++\west to ++(0pt, -9pt) to ++(9pt,0pt);
\end{tikzpicture}
\hspace{7.5pt}
\begin{tikzpicture}[baseline={(l3.base)},remember picture]
    \node[vcell=$\vl_1$] (a1) [] {};
    \node[ccell=$\vl_3$] (l3) [below=0pt of a1] {\tiny $\pair\bullet\bullet$};
    \node[vcell=$\vl_4$] (l4) [below=0pt of l3] {};
    \node [below=10pt of l3] () {$\dheapF(\vw_1' \leq \vw_1, \vw_2)(\mvarM)$};
    \draw[cptr] ($ (l3.center) + (4pt,0pt) $) to ++\east to ++\north to ++(-9pt,0pt);
    \draw[cptr] ($ (l3.center) + (-4pt,0pt) $) to ++\west to ++(0pt, -9pt) to ++(9pt,0pt);
\end{tikzpicture}
\hspace{7.5pt}
\begin{tikzpicture}[baseline={(l3.base)},remember picture]
    \node[vcell=$\vl_2$] (a2) [] {};
    \node[vcell=$\vl_1$] (a1) [below=0pt of a2] {};
    \node[ccell=$\vl_3$] (l3) [below=0pt of a1] {\tiny $\pair\bullet\bullet$};
    \node[ccell=$\vl_4$] (l4) [below=0pt of l3] {\tiny $\InjL\Unit$};
    \node [below=10pt of l3] () {$\dheapF(\vw_1, \vw_2 \leq \vw_2')(\mvarM)$};
    \draw[cptr] ($ (l3.center) + (4pt,0pt) $) to ++\east to ++\north to ++(-9pt,0pt);
    \draw[cptr] ($ (l3.center) + (-4pt,0pt) $) to ++\west to ++(0pt, -9pt) to ++(9pt,0pt);
\end{tikzpicture}
\end{center}
\end{example}

We equip $\worldC$ with data using heaplets to form the category $\instC$.
Objects in $\instC$ are worlds, and morphisms $\vw \xto\me \vw'$ consist of a morphism $\tinstU\me : \vw \to \vw'$ in $\worldC$ and a heaplet $\mvarM_\me : \dheapF(\vw' \ominus \sigma, \vw')$.
Thus $\mvarM_\me$ defines a value for every location $\vl : \vC \in \vw'$ that is not hit by $\tinstU\me$, with the value having access to the bigger $\vw'$ world.
Composition of morphisms in $\instC$ is further described in \cref{sec:worlds}, and we have a forgetful functor $\instU : \instC \to \worldC$.

The heap functor $\heapF \in \pinstC \eqdef \funC\instC\setC$ represents fully defined heaplets, given on worlds as:
\[ \heapF(\vw) = \dheapF(\vw,\vw) \iso \homF\instC\emptyset\vw \]
The functorial action on morphisms $\heapF(\me)(\mvarM)$ extends a heap $\mvarM$ with a heaplet $\mvarM_\me$, giving a covariant structure to full-ground heaps.
\begin{example}\label{ex:heap-functor}
    We form morphism $\me : (\vl_1, \vl_2 : \sortNat) \inclto \vw_2'$ in $\instC$ 
    using the heaplet from \cref{ex:overview-heaplets}, $\mvarM_\me \eqdef \dheapF(\vw_1, \vw_2 \leq \vw_2')(\mvarM)$.
    Applying $\me$ to heap $\mvarM_1$ given below results in a heap $\heapF(\me)(\mvarM_1)$ where locations $\vl_3, \vl_4$ are not reachable from $\vl_1,\vl_2$.
\begin{center}
\renewcommand\ndsize{15}
\begin{tikzpicture}[baseline={(l2.base)},remember picture]
    \node[ccell=$\vl_2$] (l2) [] {\tiny $\InjL\Unit$};
    \node[ccell=$\vl_1$] (l1) [below=0pt of l2] {\tiny $\InjR\bullet$};
    \node [below=2.5pt of l1] () {heap $\mvarM_1$};
    \draw[cptr] ($ (l1.center) + (5pt,0pt) $) to ++\east to ++\north to ++(-10pt,0);
\end{tikzpicture}
\hspace{7.5pt}
\begin{tikzpicture}[baseline={(l3.base)},remember picture]
    \node[vcell=$\vl_2$] (a2) [] {};
    \node[vcell=$\vl_1$] (a1) [below=0pt of a2] {};
    \node[ccell=$\vl_3$] (l3) [below=0pt of a1] {\tiny $\pair\bullet\bullet$};
    \node[ccell=$\vl_4$] (l4) [below=0pt of l3] {\tiny $\InjL\Unit$};
    \node [below=10pt of l3] () {heaplet $\mvarM_\me$};
    \draw[cptr] ($ (l3.center) + (4pt,0pt) $) to ++\east to ++\north to ++(-9pt,0pt);
    \draw[cptr] ($ (l3.center) + (-4pt,0pt) $) to ++\west to ++(0pt, -9pt) to ++(9pt,0pt);
\end{tikzpicture}
\hspace{7.5pt}
\begin{tikzpicture}[baseline={(l3.base)},remember picture]
    \node[ccell=$\vl_2$] (l2) [] {\tiny $\InjL\Unit$};
    \node[ccell=$\vl_1$] (l1) [below=0pt of l2] {\tiny $\InjR\bullet$};
    \node[ccell=$\vl_3$] (l3) [below=0pt of l1] {\tiny $\pair\bullet\bullet$};
    \node[ccell=$\vl_4$] (l4) [below=0pt of l3] {\tiny $\InjL\Unit$};
    \node [below=10pt of l3] () {heap $\heapF(\me)(\mvarM_1)$};
    \draw[cptr] ($ (l1.center) + (5pt,0pt) $) to ++\east to ++\north to ++(-10pt,0);
    \draw[cptr] ($ (l3.center) + (4pt,0pt) $) to ++\east to ++(0pt,6pt) to ++(-9pt,0pt);
    \draw[cptr] ($ (l3.center) + (-4pt,0pt) $) to ++\west to ++\south to ++(9pt,0pt);
\end{tikzpicture}
\end{center}
\end{example}

\paragraph{Stateful Values}
In the setting of local state, return values are not necessarily pure.
They may return an existing or newly allocated reference, and are therefore understood in the context of a world, \ie a heap layout.
Accordingly, \emph{stateful values} are an abstraction over return values that additionally track the allocations that occurred and the resulting state of the heap.

Semantically, stateful values are given via a $\pworldC$-action $(\dash \odot \heapF)$ on $\pinstC$ and a monad $(\hideMon, \eta^\hideMon, \bind^\hideMon)$ on $\pinstC$ shown in \cref{fig:fgls-construction}.
Functor $(\vX \odot \heapF)$ represents indexed sets of pairs over worlds:  it lifts $\vX \in \pworldC$ to $\pinstC$ using the forgetful functor and takes the cartesian product in $\pinstC$, $\vX \odot \heapF = \vX\tinstU \times \heapF$.
Monad $\hideMon$ on the other hand models allocation with garbage collection capabilities.

For a return type $\vX : \worldC \to \setC$, and world $\vw$ representing the existing heap layout, 
a stateful value $\eqclass{\vw \xsto\sigma \vw'}\pair{\vx'}{\mvarM'} \in \hideMon(\vX \odot \heapF)(\vw)$
allocates $\vw' \ominus \sigma$ locations,
returns value $\vx' : \vX(\vw')$ and updates the heap to $\mvarM' : \heapF(\vw')$.
These stateful values are viewed only up to observable allocations, thus modelling garbage collection.
\begin{example}\label{ex:overview-stateful-values}
    All four stateful values are over public world $(\vl_1,\vl_2 : \sortNat)$ with $\sem{\tyRef\sortNat}$ return value:
\begin{center}
\renewcommand\ndsize{15}
\begin{tikzpicture}[baseline={(l2.base)},remember picture]
    \node[ccell=$\vl_2$] (l2) [] {\tiny $\InjR\bullet$};
    \node[ccell=$\vl_1$] (l1) [below=0pt of l2] {\tiny $\InjR\bullet$};
    \node [below=2.5pt of l1] () {$\eqclass{\vl_1,\vl_2 \leq \vl_1, \vl_2 : \sortNat}\pair{\vl_1}{\mvarM_\cycle}$ };
    \draw[cptr] ($ (l1.center) + (5pt,0pt) $) to ++(10pt,0pt) to ++(0,6pt) to ++(-8pt,0);
    \draw[cptr] ($ (l2.center) + (5pt,0pt) $) to ++(12pt,0pt) to ++(0,-10pt) to ++(-10pt,0);
    \draw[|->, black] ($ (l1.west) + (-16pt, 1pt) $) -- ($ (l1.west) + (-7pt, 1pt) $);
\end{tikzpicture}
\hspace{1pt} $\neq$ \hspace{1pt}
\begin{tikzpicture}[baseline={(l2.base)},remember picture]
    \node[ccell=$\vl_2$] (l2) [] {\tiny $\InjL\Unit$};
    \node[ccell=$\vl_1$] (l1) [below=0pt of l2] {\tiny $\InjR\bullet$};
    \node [below=2.5pt of l1] () {$\eqclass{\vl_1,\vl_2 \leq \vl_1, \vl_2 : \sortNat}\pair{\vl_1}{\mvarM_1}$ };
    \draw[cptr] ($ (l1.center) + (5pt,0pt) $) to ++\east to ++\north to ++(-10pt,0);
    \draw[|->, black] ($ (l1.west) + (-16pt, 1pt) $) -- ($ (l1.west) + (-7pt, 1pt) $);
\end{tikzpicture}
\hspace{1pt} = \hspace{1pt}
\begin{tikzpicture}[baseline={(l3.base)},remember picture]
    \node[ccell=$\vl_2$] (l2) [] {\tiny $\InjL\Unit$};
    \node[ccell=$\vl_1$] (l1) [below=0pt of l2] {\tiny $\InjR\bullet$};
    \node[ccell=$\vl_3$] (l3) [below=0pt of l1] {\tiny $\pair\bullet\bullet$};
    \node[ccell=$\vl_4$] (l4) [below=0pt of l3] {\tiny $\InjL\Unit$};
    \node [below=10pt of l3] () {$\eqclass{\vl_1,\vl_2 \leq \vw_2'}\pair{\vl_1}{\heapF(\me)(\mvarM_1)}$};
    \draw[cptr] ($ (l1.center) + (5pt,0pt) $) to ++\east to ++\north to ++(-10pt,0);
    \draw[cptr] ($ (l3.center) + (4pt,0pt) $) to ++\east to ++(0pt,6pt) to ++(-9pt,0pt);
    \draw[cptr] ($ (l3.center) + (-4pt,0pt) $) to ++\west to ++\south to ++(9pt,0pt);
    \draw[|->, black] ($ (l1.west) + (-16pt, 1pt) $) -- ($ (l1.west) + (-7pt, 1pt) $);
\end{tikzpicture}
\hspace{1pt} $\neq$ \hspace{1pt}
\begin{tikzpicture}[baseline={(l2.base)}]
    \node[ccell=$\vl_3$] (l3) [] {\tiny $\InjL\Unit$};
    \node[ccell=$\vl_2$] (l2) [below=0pt of l3] {\tiny $\InjR\bullet$};
    \node[ccell=$\vl_1$] (l1) [below=0pt of l2] {\tiny $\InjR\bullet$};
    \node [below=2.5pt of l1] () {$\eqclass{\vl_1,\vl_2 \leq (\vl_1,\vl_2,\vl_3 : \sortNat)}\pair{\vl_1}{\mvarM_2}$};
    \draw[cptr] ($ (l1.center) + (5pt,0pt) $) to ++\east to ++(0,6pt) to ++(-10pt,0);
    \draw[cptr] ($ (l2.center) + (5pt,0pt) $) to ++\east to ++\north to ++(-10pt,0);
    \draw[|->, black] ($ (l1.west) + (-16pt, 1pt) $) -- ($ (l1.west) + (-7pt, 1pt) $);
\end{tikzpicture}
\end{center}
    Only the middle two stateful values are equal because allocations $\vl_3, \vl_4$ is not observable: they are not returned nor are they reachable from $\vl_1, \vl_2$
    and are therefore garbage collected.
    In contrast, the left-most stateful value stores $\InjR\Unit$ at $\vl_2$ and in the right-most, allocation $\vl_3$ is observable from $\vl_1$.
\end{example}
\noindent
Intuitively, unit $\mvarM^\hideMon$ injects into a context with no allocations and bind $\bind^\hideMon$ composes the allocations.
\begin{figure}[ht]
\vspace{-15pt}
\(\begin{tikzcd}
    {\pworldC} \arrow[rr, "\dash \odot \heapF", shift left=3] 
        & \perp & {\pinstC} \arrow[ll, "\heapF \multimap \dash", shift left=3] \arrow[loop right, r, "\hideMon"]
\end{tikzcd}
\qquad
\begin{array}{r@{\;}c@{\;}l}
\monFLS & : & \pworldC \to \pworldC \\
\monFLS\vX & = & \heapF \multimap \hideMon(\vX \odot \heapF)
\end{array}\)
\vspace{-10pt}
\Description{State-transformer description of the FGLS Monad}
\caption{State-transformer description of the FGLS Monad}
\label{fig:fgls-construction}
\vspace{-10pt}
\end{figure}

\paragraph{The Monad}
Action $(\dash \odot \heapF)$ has a right adjoint $(\heapF \multimap \dash)$ by generalities, inducing a state transformer over full ground heaps~(\cref{fig:fgls-construction}).
We think of $(\heapF \multimap \vA)$ as lifting the exponential $\Exp\vA\heapF$ in $\pinstC$ to $\pworldC$.
More concretely, at world $\vw$ it is equivalent to:
\[ (\heapF \multimap \vA)(\vw) \iso \homF\pinstC{\homF\worldC\vw\dash \odot \heapF}\vA \]
Specialising this adjunction to stateful values gives us the function-like structure of $\monFLS$:
\[\begin{tikzcd}
    {\homF{\pinstC}{\vX \odot \heapF}{\hideMon(\vY \odot \heapF)}}
    \ar[r, bend left=10, anchor=north, "\curry_{\vX,\vY}"] \ar[r, phantom, "\iso"]
    &
    {\homF{\pworldC}\vX{\heapF \multimap \hideMon(\vY \odot \heapF)}} \ar[l, anchor=south, bend left=10, "\uncurry_{\vX,\vY}" yshift=-2pt]
\end{tikzcd}
\qquad
    \eval \eqdef \uncurry_{\monFLS\vX, \vX}(1_{\monFLS\vX})
\]
Therefore, a FGLS computation $\phi \in \monFLS\vX\vw$ is a natural transformation $\homF\worldC\vw\dash \odot \heapF \to \hideMon(\vX \odot \heapF)$:
defines a stateful value $\phi_{\vw'}\pair{\sigma'}{\mvarM'}$ for every $\vw \xto\sigma \vw' \in \worldC$ and heap $\mvarM' : \heapF(\vw')$.
Naturality says that $\phi$ must only examine the reachable parts of the heap starting from world $\vw$: different unreachable heap segments should not be observable.
Finally, the monad structure of $\monFLS$ is given by the monad $\hideMon$ and $\curry/\uncurry$:
\[\begin{array}{r@{\;}l@{\;}l}
    \eta_\vX & \eqdef \curry_{\vX, \vX}(\eta^\hideMon_{\vX \odot \heapF}) & : \vX \to \monFLS\vX \\[1ex]
    \bind(\vf : \vX \to \monFLS\vY) & \eqdef \heapF \multimap (\bind^\hideMon(\uncurry_{\vX,\vY}\vf)) & : \monFLS\vX \to \monFLS\vY \
\end{array}\]
$\eta_\vX(\vx)$ is a computation that returns $\vx$ leaving the heap unchanged (no allocations occur),
and $\vc\bind\vf$ first evaluates $\vc$ and calls $\vf$ with its result value and the modified heap.

We note in passing, that the FGLS monad is equivalently given using an End:
\[ \monFLS\vX\vw \iso \End{\vw \injto \vw' \in \commaC\vw\instU} \heapF\vw' \to \hideMon(\vX \odot \heapF)(\vw') \]
where $\commaC\vw\instU$ is the comma category over world $\vw \in \objC\worldC$ and forgetful functor $\tinstU : \instC \to \worldC$. 




\subsection{Non Finitariness of FGLS Monad}\label{sec:overview-non-finitariness}

\def\heapSize{\mathit{heapSize}}
\def\getOp{\mathit{get}}
\def\putF{\mathit{put}_\False}
\def\putT{\mathit{put}_\True}
\def\GS{\mathit{GS}}

As explained in the introduction, one of the long-term goals of our work is finding an equational characterization of FGLS via a suitable notion of algebraic theory.
To be able to use the framework of parameterized algebraic theories~\cite{staton2013instances} for this purpose, which provides a syntax for presenting theories, or that of enriched Lawvere theories~\cite{power1999enriched}, we would require the FGLS monad to be \emph{finitary}.
In this section we discuss what being finitary means and why the FGLS monad does not satisfy this requirement.


We explain what it means for a monad to be finitary through the simpler setting of the global state monad on $\setC$.
Let $\vS$ be a set that represents the state, which may be a single ground value cell (\eg $\vS \eqdef \typeBool, \typeNat$), or an entire heap (\eg $\vS \eqdef \typeLoc \finpto \typeVal$).
The global state monad is given as $\GS_\vS(\vX) = \vS \to \vX \times \vS$,
which represents state transforming functions with return values of type $\vX$.

When $\vS$ is finite, the global state monad is finitary.
For instance, $\GS_\typeBool$ can be equivalently described via an algebraic theory that has operations $\getOp(\vx,\vy), \putF(\vx), \putT(\vx)$
where $\getOp$ selects continuation $\vx$ or $\vy$ according to the contents of the cell, and $\mathit{put}_\vb$ updates the value of the store.
However for non-finite $\vS$, $\GS_\vS$ is not finitary: $\getOp$ would need a countable arity $\typeNat$ to represent state $\vS \eqdef \typeNat$.
Put differently, $\GS_\vS$ is finitary if for every $\vX, \vf \in \GS_\vS\vX$, $(\proj_1 \comp \vf)$ has finite image~\cite{kelly1982structures}.

To motivate why finitariness fails for the FGLS monad, consider $\vS \eqdef \typeLoc \finpto \typeVal$ representing the heaps from \cref{sec:lambdaref-os}.
The global state monad $\GS_{\typeLoc \finpto \typeVal}$ is not finitary for the same reasons as above:
there exists a function $\Lam\mvarM \card{\supp\mvarM} \in \GS_{\typeLoc \finpto \typeVal}(\typeNat)$ that counts the number of defined locations and has a countable image.

We define an analogous function in the FGLS setting, called $\heapSize$.
First let the result type be pure $\typeNat$ values given by the constant functor $\Delta_\typeNat = \vw \mapsto \typeNat \in \pworldC$.
We define $\heapSize$ at world $(\vl : \sortNat)$ representing a single public location of sort $\sortNat$ pointing to the heap.\footnote{Not all interpretation functions yield a non-finitary FGLS monad, but examples of interest that support cyclic data structures will be non-finitary.}

Recall from (\cref{sec:overview-fgls}), to give $\heapSize \in \monFLS(\Delta_\typeNat)(\vl : \typeNat)$ is to give a natural transformation $\homF\worldC{\vl : \sortNat}\dash \odot \heapF \to \hideMon(\Delta_\typeNat \odot \heapF)$, where naturality says that $\heapSize$ must only examine the reachable parts of the heap starting from $\vl$.
Therefore, we define $\heapSize$ at $(\vl : \sortNat) \xto\sigma \vw$, and heap $\mvarM : \heapF(\vw)$ by traversing $\mvarM$ starting from $\sigma(\vl)$ and returning $\mvarM^\hideMon\pair\vn\mvarM$ where $\vn$ is the number of distinct locations it has encountered.
For instance $\heapSize$ returns for the following heaps:
\begin{center}
\renewcommand\ndsize{15}
\begin{tikzpicture}[baseline={(l1.base)},remember picture]
    \node[ccell=$\vl$] (l1) [] {\tiny $\InjR\bullet$};
    \node [below=5pt of l1] () {$\heapSize(\vl) = 1$};
    \draw[cptr] ($ (l1.center) + (5pt,0pt) $) to ++\east to ++\north to ++\west to ++(0,-4pt);
\end{tikzpicture}
\hspace{7.5pt}
\begin{tikzpicture}[baseline={(l1.base)},remember picture]
    \node[ccell=$\vl$] (l1) [] {\tiny $\InjR\bullet$};
    \node[ccell] (l2) [above=0pt of l1] {\tiny $\InjL\Unit$};
    \node [below=5pt of l1] () {$\heapSize(\vl) = 2$};
    \draw[cptr] ($ (l1.center) + (5pt,0pt) $) to ++\east to ++\north to ++(-10pt,0);
\end{tikzpicture}
\hspace{7.5pt}
\begin{tikzpicture}[baseline={(l3.base)},remember picture]
    \node[ccell=$\vl$] (l1) [] {\tiny $\InjR\bullet$};
    \node[ccell] (l2) [above=0pt of l1] {\tiny $\InjL\Unit$};
    \node[ccell=$\vl'$] (l3) [below=0pt of l1] {\tiny $\InjR\bullet$};
    \node [below=5pt of l3] () {$\heapSize(\vl) = 2$};
    \draw[cptr] ($ (l1.center) + (5pt,0pt) $) to ++\east to ++\north to ++(-10pt,0);
    \draw[cptr] ($ (l3.center) + (5pt,0pt) $) to ++\east to ++(0,6pt) to ++(-10pt,0);
\end{tikzpicture}
\end{center}
Since $\mvarM$ is finite, $\heapSize$ is well-founded.
With this example we will show that the FGLS monad is not finitary as it returns countably-many different values according to the heap it is applied to.
We refer the reader to (\cref{lem:non-finitariness-proof}) for the formal definition of $\heapSize$.

We note that even though finitariness for the FGLS monad may seem restrictive as it places bounds on heap traversals,
a programming language that supports complete heap traversals while remaining in the realm of terminating languages is also strange.
One might imagine introducing $\heapSize$ as a built-in operation in the language but that operation either exposes unnecessary lower level details about the runtime or requires an expensive garbage collection operation.


In the next section, and in \S4, we develop a finitary version of the monad that can in the future be used as a basis for an equational axiomatization of FGLS.
We construct our finitary monad by removing non-finitary computations such as $\heapSize$ from the original monad.

\subsection{Finitary FGLS Monad}\label{sec:overview-fin-fgls}

The central idea behind finitary computations is that they return finitely many values. 
Applying this to the FGLS setting means that it cannot return a different value for every possible heap.
We therefore represent a finitary FGLS computation by finitely partitioning the heap space into segments and assigning a stateful value to each segment.
We begin by elaborating on the idea of finite partitions of the heap space.

\emph{Templates} are two-kinded worlds: one kind for concrete locations and another for location variables.
While concrete locations represent cells with defined data, location variables are extension points which we can instantiate to extend the heap.
We write templates as $\tau = \vl_1 : \vC_1, \va_2 : \vC_2, \vl_3 : \vC_3$, where metavariable $\tau$ ranges over templates, $\vl$ over concrete locations and $\va$ over location variables.
We write $\concL\tau$ for the world with only concrete locations $\vl_1 : \vC_1, \vl_3 : \vC_3$ and $\varL\tau = \va_2 : \vC_2$ for the world with only the variables.
We further write $\anyL\tau$ for the template seen as a world, that is, we forget the distinction given by the two kinds. 

\emph{Stencils} represent collections of heaps with a common set of public locations.
That is, for a public world $\vw$ and template $\tau$, a stencil is a pair $\pair{\sigma_\vs}{\mvarM_\vs} \in \stencilF\vw\tau$,
with $\sigma_\vs : \vw \to \anyL\tau \in \worldC$ giving the mapping of the public locations and heaplet $\mvarM_\vs : \dheapF(\concL\tau, \anyL\tau)$ the possibly incomplete segment of the heap.
We equip stencils with a covariant structure that extends the heap by instantiating variables $\varL\tau$. 
Thus, a stencil can be thought of as a collection of concrete heaps given by all its extensions.

\def\stenZ{\vs_0}
\def\stenS{\vs_{\mathsf{S}}}
\def\stenUno{\vs_{1}}
\begin{example}
We define three stencils all over public world $(\vl : \sortNat)$:
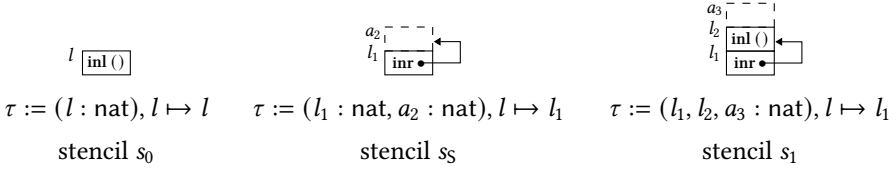
\begin{center}
\renewcommand\ndsize{18}
\begin{tikzpicture}
    \node[ccell=$\vl$] (l) [] {\tiny $\,\InjL\Unit\,$};
    \node [below=5pt of l] (Ztext) {$\tau \eqdef (\vl : \sortNat), \vl \mapsto \vl$};
    \node [below=0pt of Ztext] () {stencil $\stenZ$};
\end{tikzpicture}
\hspace{5pt}
\begin{tikzpicture}
    \node[vcell=$\va_2$] (a2) [] {};
    \node[ccell=$\vl_1$] (l1) [below=0pt of a2] {\tiny $\InjR\bullet$}; 
    \draw[cptr] ($ (l1.center) + (5pt,0pt) $) to ++\east to ++(0,8pt) to ++(-10pt,0);
    \node [below=5pt of l1] (Stext) {$\tau \eqdef (\vl_1 : \sortNat, \va_2 : \sortNat), \vl \mapsto \vl_1$};
    \node [below=0pt of Stext] () {stencil $\stenS$};
\end{tikzpicture}
\hspace{5pt}
\begin{tikzpicture}
    \node[vcell=$\va_3$] (a3) [] {};
    \node[ccell=$\vl_2$] (l2) [below=0pt of a3] {\tiny $\,\InjL\Unit\,$};
    \node[ccell=$\vl_1$] (l1) [below=0pt of l2] {\tiny $\InjR\bullet$}; 
    \draw[cptr] ($ (l1.center) + (5pt,0pt) $) to ++\east to ++(0,8pt) to ++(-10pt,0);
    \node [below=5pt of l1] (1text) {$\tau \eqdef (\vl_1,\vl_2, \va_3 : \sortNat), \vl \mapsto \vl_1$};
    \node [below=0pt of 1text] () {stencil $\stenUno$};
\end{tikzpicture}
\end{center}
Stencil $\stenS$ represents heaps that are defined on at least two locations and have $(\InjR\bullet)$ stored in $\vl$,
\eg the first three heaps below are valid extensions whereas the right-most is not.
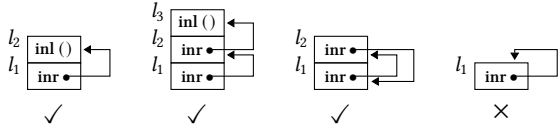
\begin{center}
\renewcommand\ndsize{20}
\begin{tikzpicture}[baseline={(l1.base)}]
    \node[ccell=$\vl_2$] (l2) [] {\tiny $\InjL\Unit$};
    \node[ccell=$\vl_1$] (l1) [below=0pt of l2] {\tiny $\InjR\bullet$};
    \node[below=2pt of l1] () {$\checkmark$};
    \draw[cptr] ($ (l1.center) + (5pt,0pt) $) to ++\east to ++\north to ++(-10pt,0);
\end{tikzpicture}
\hspace{7.5pt}
\begin{tikzpicture}[baseline={(l1.base)}]
    \node[ccell=$\vl_3$] (l3) [] {\tiny $\InjL\Unit$};
    \node[ccell=$\vl_2$] (l2) [below=0pt of l3] {\tiny $\InjR\bullet$};
    \node[ccell=$\vl_1$] (l1) [below=0pt of l2] {\tiny $\InjR\bullet$};
    \node[below=2pt of l1] () {$\checkmark$};
    \draw[cptr] ($ (l2.center) + (5pt,0pt) $) to ++\east to ++\north to ++(-10pt,0);
    \draw[cptr] ($ (l1.center) + (5pt,0pt) $) to ++\east to ++(0,8pt) to ++(-10pt,0);
\end{tikzpicture}
\hspace{7.5pt}
\begin{tikzpicture}[baseline={(l1.base)}]
    \node[ccell=$\vl_2$] (l2) [] {\tiny $\InjR\bullet$};
    \node[ccell=$\vl_1$] (l1) [below=0pt of l2] {\tiny $\InjR\bullet$};
    \node[below=2pt of l1] () {$\checkmark$};
    \draw[cptr] ($ (l1.center) + (5pt,0pt) $) to ++\east to ++(0,8pt) to ++(-10pt,0);
    \draw[cptr] ($ (l2.center) + (5pt,0pt) $) to ++\eeast to ++\south to ++(0pt,-2pt) to ++(-16pt,0);
\end{tikzpicture}
\hspace{7.5pt}
\begin{tikzpicture}[baseline={(l1.base)}]
    \node[ccell=$\vl_1$] (l1) [] {\tiny $\InjR\bullet$};
    \node[below=2pt of l1] () {$\times$};
    \draw[cptr] ($ (l1.center) + (5pt,0pt) $) to ++\east to ++\north to ++\west to ++(0,-4pt);
\end{tikzpicture}
\end{center}
\end{example}
A stencil is said to be \emph{connected} if every location can be reached by traversing the heaplet starting from a public location. 
In the above example, stencils $\stenZ,\stenS$ are connected, whereas $\stenUno$ is not because $\va_3$ is not reachable from $\vl_1$.

For a finite set of templates $\Sigma$, we call a $\Sigma$-indexed family of stencils over a public world $\vw$ a \emph{cover}. 
We refer to \emph{connected covers} if all its stencils are connected, 
\emph{non-overlapping covers} if the collections of heaps represented by the stencils are disjoint and 
\emph{complete covers} if every heap is included by exactly one stencil in the cover. 
Completeness of a cover implies that it is both connected and non-overlapping.

\def\stenUno{\vs_{1}}
\begin{example}
    The cover $\pair\stenZ\stenS$ over public world $(\vl : \sortNat)$ is complete since every heap $\mvarM$ with a public location $(\vl : \sortNat)$ will store in it either $\InjL\Unit$ (covered by $\stenZ$) or $\InjR{\vl_2}$ (covered by $\stenS$).
    On the other hand, the cover $\triple\stenZ\stenS\stenUno$, is not complete since it is neither connected (due to $\stenUno$) nor is it non-overlapping: 
    heap $\mvarM_2$ from \cref{ex:overview-stateful-values} is included by both $\stenS$ and $\stenUno$.
\end{example}
A complete cover is thus a finite partition of the heap space over a common set of public locations.
\Cref{sec:theory-templates} introduces templates formally and in \cref{sec:covers} we assemble them to form stencils and covers.

\emph{Resultlets} are a means to assign stateful values to each stencil in a cover.
Since stencils represent collections of heaps via the location variables, we similarly introduce variables to stateful values and talk about collections of them. 

An $(\vX \in \pworldC)$-valued resultlet over template $\tau$ is thought of as a stateful value that has access to variables $\varL\tau$. 
Specifically, a resultlet $\inj_{\concL\tau \xsto{\sigma'} \vw'} \pair{\vx'}{\mvarM'} \in \typeRes\vX\tau$ has the following data:
\begin{itemize}
    \item An injection $\concL\tau \xto{\sigma'} \vw' \in \worldC$ representing an allocation of $\vw' \ominus \sigma'$ locations.
    \item An $\vX$-valued return value $\vx' : \vX(\vw' \oplus \varL\tau)$ that can also return variables $\varL\tau$.
    \item A heaplet $\mvarM' : \dheapF(\vw', \vw' \oplus \varL\tau)$ that is defined on $\vw'$ and has access to variables $\varL\tau$.
\end{itemize}
Notice that resultlets cannot instantiate variables $\varL\tau$ in their data: these variables represent the state of the heap that is left unchanged.
Just as with stencils, resultlets are thought of as collections of stateful values given by all the instantiations of $\varL\tau$. 

\def\resC{\vr_\cycle}
\def\resS{\vr_\Succ}
\begin{example}
\renewcommand\ndsize{15}
We create two $\sem{\tyRef\sortNat}$-valued resultlets over templates $(\vl : \sortNat)$ and $(\vl_1, \va_2 : \sortNat)$.
Resultlet $\resC$ allocates a new cell $\vl'$, creates a cycle with $\vl$ and returns $\vl$, whereas $\resS$ does not allocate any new cells but instead only writes $\InjR{\va_2}$ to $\vl_1$.
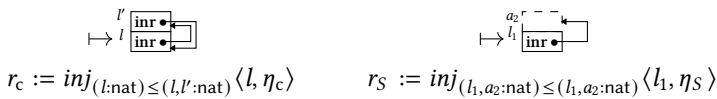
\begin{center}
\begin{tikzpicture}[baseline={(l.base)}]
    \node[ccell=$\vl'$] (l') [] {\tiny $\InjR\bullet$};
    \node[ccell=$\vl$] (l) [below=0pt of l'] {\tiny $\InjR\bullet$};
    \node [below=2.5pt of l] () {$\resC \eqdef \inj_{(\vl : \sortNat) \leq (\vl, \vl': \sortNat)}\pair\vl{\mvarM_\cycle}$};
    \draw[|->, black] ($ (l.west) + (-16pt, 1pt) $) -- ($ (l.west) + (-7pt, 1pt) $);
    \draw[cptr] ($ (l.center) + (5pt,0pt) $) to ++(10pt,0pt) to ++(0,6pt) to ++(-8pt,0);
    \draw[cptr] ($ (l'.center) + (5pt,0pt) $) to ++(12pt,0pt) to ++(0,-10pt) to ++(-10pt,0);
\end{tikzpicture}
\hspace{15pt}
\begin{tikzpicture}[baseline={(l1.base)}]
    \node[vcell=$\va_2$] (a2) [] {};
    \node[ccell=$\vl_1$] (l1) [below=0pt of a2] {\tiny $\InjR\bullet$};
    \node [below=2.5pt of l1] () {$\resS \eqdef \inj_{(\vl_1,\va_2 : \sortNat) \leq (\vl_1, \va_2 : \sortNat)}\pair{\vl_1}{\mvarM_\Succ}$ };
    \draw[cptr] ($ (l1.center) + (5pt,0pt) $) to ++\east to ++\north to ++(-10pt,0);
    \draw[|->, black] ($ (l1.west) + (-16pt, 1pt) $) -- ($ (l1.west) + (-7pt, 1pt) $);
\end{tikzpicture}
\end{center}
Recall the four stateful values from \cref{ex:overview-stateful-values}.
Resultlet $\resS$ can be extended to the right-most three using data
\begin{tikzpicture}
    \node[ccell=$\vl_2$] (l2) [] {\tiny $\InjL\Unit$};
\end{tikzpicture}\hspace{2pt}
and
\begin{tikzpicture}[baseline={(l2.base)}]
    \node[ccell=$\vl_3$] (l3) [] {\tiny $\InjL\Unit$};
    \node[ccell=$\vl_2$] (l2) [below=0pt of l3] {\tiny $\InjR\bullet$};
    \draw[cptr] ($ (l2.center) + (5pt,0pt) $) to ++\east to ++\north to ++(-10pt,0);
\end{tikzpicture}.
In contrast, $\resC$ cannot be extended to the left-most stateful value since $\vl_2$ is an existing location whereas $\vl'$ has been allocated.
\end{example}
For a finite set of templates $\Sigma$, we call a $\Sigma$-indexed family of $\vX$-valued resultlets an \emph{evaluation}. 
Resultlets and evaluations are further described in \cref{sec:evaluation}.

\paragraph{Finitary Computations}
We link covers with evaluations to describe the finitary FGLS computations.
Every triple consisting of a finite set of templates $\Sigma$, a $\Sigma$-indexed complete cover over $\vw$ and $\Sigma$-indexed $\vX$-valued evaluation 
describes a FGLS computation in $\monFLS\vX\vw$ in a non-unique way.

For instance, $[(\vl : \sortNat) \mapsto \stenZ, \resC; (\vl_1, \va_2 : \sortNat) \mapsto \stenS, \resS]$ represents a finitary FGLS computation over public world $(\vl : \sortNat)$ and templates $(\vl : \sortNat)$ and $(\vl_1, \va_2 : \sortNat)$.
The computation that it realises is defined at $(\vl : \sortNat) \xto\sigma \vw', \mvarM' : \heapF(\vw')$, by case splitting on the contents of cell $\proj_{\sigma(\vl)} \mvarM'$. 
If it is $\InjL\Unit$, it overwrites it with the cycle given by $\resC$ leaving the rest of the heap unchanged.
If it is $\InjR{\vl_2}$, for some location $\vl_2 \in \vw'$, it leaves the whole heap $\mvarM'$ unchanged.
In both cases $\sigma(\vl)$ is returned.

We say that $\phi \in \monFLS\vX\vw$ is finitary if it can be expressed using covers and evaluations.
In \cref{sec:fin-monad} we show that the finitary computations are closed under the functor and monad structure and therefore define a submonad of $\monFLS$. 


\section{Worlds, Templates and Instantiations}\label{sec:theory-templates}

In this section we introduce the machinery needed to describe the finitary FGLS monad.
We first revisit worlds and instantiations~(\cref{sec:worlds}) from \cite{kammar2017monad} and then extend worlds to two-kinded worlds called \emph{templates}~(\cref{sec:tempC}).
We then introduce instantiation data to templates~(\cref{sec:tinstC}) and lastly develop the theory for collections of templates with instantiation data~(\cref{sec:famC-tinstC}).  

\subsection{Worlds and Instantiations}\label{sec:worlds}

We first revisit categories $\worldC$ and $\instC$ introduced in \cref{sec:overview-fgls}.

Recall that category $\worldC$ has worlds $\vw$ as objects and morphisms are sort-preserving injections $\vw \xto\sigma \vw'$. 
We can form the disjoint union of two worlds denoted by $\vw_1 \oplus \vw_2$ that has $\card{\vw_1} + \card{\vw_2}$ many locations.
This construction extends to a functor $\oplus : \worldC \times \worldC \to \worldC$ together with injections $\incl_1 : \vw_1 \to \vw_1 \oplus \vw_2$ and $\incl_2 : \vw_2 \to \vw_1 \oplus \vw_2$ but note that it does not satisfy the universal property of coproducts because of injectivity of morphisms in $\worldC$.

We equip $\worldC$ with instantiation data to form category $\instC$, 
which has worlds $\vw$ as objects, and morphisms $\vw \xto\me \vw'$ consist of a sort preserving injection $\tinstU\me$ and heaplet
\[ \mvarM_\me : \dheapF(\vw' \ominus \tinstU\me, \vw') = \Prod{\vl : \vC \in \vw' \ominus \tinstU\me} \sem{\ctype\vC}(\vw') \]
Semantic full-ground values $\sem{\ctype\vC}$ are given for reference types using the hom-functor $\sem{\tyRef\vC} = \homF\worldC{(\vl : \vC)}\dash$,
and thus $\sigma \in \sem{\tyRef\vC}(\vw)$ is a morphism $(\vl : \vC) \xto\sigma \vw$ in $\worldC$, \ie a selection of a single $\vC$-sorted location in $\vw$.
The $\tyProd, \tySum$ types are given via the cartesian closed structure of $\pworldC$.
Note that semantic values are in one-to-one correspondence with syntactic closed full ground values: $\sem{\ctype\vC}(\vw) \iso \typedv\ctxNil\vw{\ctype\vC}$.

Composition of $\vw_1 \xto{\me_1} \vw_2 \xto{\me_2} \vw_3$ is defined by composition of the underlying injection and composition of the instantiation data 
$\mvarM_{\me_2 \comp \me_1} : \dheapF(\vw_3 \ominus (\tinstU{\me_2} \comp \tinstU{\me_1}), \vw_3)$ is given by
\[
    \vl : \vC \in \vw_3 \ominus (\tinstU{\me_2} \comp \tinstU{\me_1}) 
             \mapsto \begin{cases}
                            \sem{\ctype\vC}(\tinstU{\me_2})(\proj_{\vl'}\mvarM_{\me_1}) & \mif \tinstU{\me_2}(\vl') = \vl \\
                            \proj_\vl{\mvarM_{\me_2}} & \mif \vl \in \vw_3 \ominus \tinstU{\me_2}
                        \end{cases}
\]
A $\vl : \vC \in \vw_3 \ominus (\tinstU{\me_2} \comp \tinstU{\me_1})$ comes either from $\vw_2 \ominus \tinstU{\me_1}$ or from $\vw_3 \ominus \tinstU{\me_2}$, 
and $\sem{\ctype\vC}(\tinstU{\me_2})$ promotes locations $\vw_2$ to $\vw_3$.
\begin{example}
    The composition of $\vl_1 : \sortNat \xinclto{\me_1} (\vl_1, \vl_2, \vl_3 : \sortNat) \xinclto{\me_2} (\vl_1, \vl_2, \vl_3, \vl_4, \vl_5 : \sortNat)$ is given as
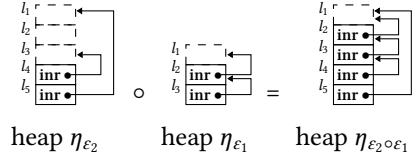
\begin{center}
\renewcommand\ndsize{15}
\begin{tikzpicture}[baseline={(l5.base)},remember picture]
    \node[vcell=$\vl_1$] (l1) [] {};
    \node[vcell=$\vl_2$] (l2) [below=0pt of l1] {};
    \node[vcell=$\vl_3$] (l3) [below=0pt of l2] {};
    \node[ccell=$\vl_4$] (l4) [below=0pt of l3] {\tiny $\InjR\bullet$};
    \node[ccell=$\vl_5$] (l5) [below=0pt of l4] {\tiny $\InjR\bullet$};
    \node [below=5pt of l5] () {heap $\mvarM_{\me_2}$};
    \draw[cptr] ($ (l4.center) + (5pt,0pt) $) to ++\east to ++\north to ++(-9pt,0);
    \draw[cptr] ($ (l5.center) + (5pt,0pt) $) to ++\eeast to ++\north to ++\north to ++\north to ++\nnorth to ++(-14pt,0);
\end{tikzpicture}
$\;\,\comp$
\begin{tikzpicture}[baseline={(l3.base)},remember picture]
    \node[vcell=$\vl_1$] (l1) [] {};
    \node[ccell=$\vl_2$] (l2) [below=0pt of l1] {\tiny $\InjR\bullet$};
    \node[ccell=$\vl_3$] (l3) [below=0pt of l2] {\tiny $\InjR\bullet$};
    \node [below=5pt of l3] () {heap $\mvarM_{\me_1}$};
    \draw[cptr] ($ (l2.center) + (5pt,0pt) $) to ++\east to ++\north to ++(-9pt,0);
    \draw[cptr] ($ (l3.center) + (5pt,0pt) $) to ++\east to ++(0,6pt) to ++(-9pt,0);
\end{tikzpicture}
=
\begin{tikzpicture}[baseline={(l5.base)},remember picture]
    \node[vcell=$\vl_1$] (l1) [] {};
    \node[ccell=$\vl_2$] (l2) [below=0pt of l1] {\tiny $\InjR\bullet$};
    \node[ccell=$\vl_3$] (l3) [below=0pt of l2] {\tiny $\InjR\bullet$};
    \node[ccell=$\vl_4$] (l4) [below=0pt of l3] {\tiny $\InjR\bullet$};
    \node[ccell=$\vl_5$] (l5) [below=0pt of l4] {\tiny $\InjR\bullet$};
    \node [below=5pt of l5] () {heap $\mvarM_{\me_2 \comp \me_1}$};
    \draw[cptr] ($ (l2.center) + (5pt,0pt) $) to ++\east to ++(0,6pt) to ++(-9pt,0);
    \draw[cptr] ($ (l3.center) + (5pt,0pt) $) to ++\east to ++(0,6pt) to ++(-9pt,0);
    \draw[cptr] ($ (l4.center) + (5pt,0pt) $) to ++\east to ++(0,6pt) to ++(-9pt,0);
    \draw[cptr] ($ (l5.center) + (5pt,0pt) $) to ++\eeast to ++\north to ++\north to ++\north to ++\nnorth to ++(-14pt,0);
\end{tikzpicture}
\end{center}
\end{example}
\noindent
$\worldC$ has pullbacks 
\begin{tikzcd}[sep=small]
    {\vw_1 \pltimes\vw \vw_2} \ar[r, "\proj_2" xshift=-2pt]  \ar[d, "\proj_1"' xshift=-2pt] \pbsquare{dr}
        & {\vw_2} \ar[d, "\sigma_2" xshift=1pt]
    \\
    {\vw_1} \ar[r, "\sigma_1"']
        & {\vw}
\end{tikzcd}
given by image intersection: $\supp{\vw_1 \pltimes\vw \vw_2} = \img{\sigma_1} \setInt \img{\sigma_2}$.

\subsection[Templates]{Templates}\label{sec:tempC}

We introduce worlds with two kinds of locations
(i) concrete locations that represent pointers to allocated data and 
(ii) location variables that stand for extension points that hold \emph{arbitrary cell contents}.
These two-kinded worlds, called \emph{templates}, are modelled using category $\tempC$ given by
\begin{itemize}
\item \emph{Objects:} A world $\Plus\vw$ with a chosen subset $\Minus\vw \inclto \Plus\vw$.
\item \emph{Morphisms:} A morphism $(\Minus\vw \inclto \Plus\vw) \xto\rho (\Minus{\hat\vw} \inclto \Plus{\hat\vw})$ 
                        is a sort-preserving injection $\Plus\vw \xto{\Plus\rho} \Plus{\hat\vw} \in \worldC$
                        that preserves the chosen subset, that is, the restriction of $\Plus\rho$ to $\Minus\vw$ factors through $\Minus{\hat\vw}$:
$\begin{tikzcd}
    {\Minus\vw} \ar[rr, bend right=20, "\Plus\rho\dres|{\Minus\vw}"'] \ar[r, dashed, "\exists"] & {\Minus{\hat\vw}} \ar[r, hook] & \Plus{\hat\vw} 
\end{tikzcd}$
\end{itemize}
Composition is given by composition of the underlying injection.
We have two forgetful functors $\concL\dash, \anyL\dash : \tempC \to \worldC$, projecting the $-$ and $+$ components respectively:
 $\concL{\Minus\vw \inclto \Plus\vw} = \Minus\vw, \anyL{\Minus\vw \inclto \Plus\vw} = \Plus\vw$.
We use metavariable $\tau$ to range over objects (templates) in $\tempC$, and call $\concL\tau$ its concrete locations, 
$\anyL\tau$ its combined locations, and $\varL\tau \eqdef \anyL\tau \setDiff \concL\tau$ its location variables.
If it is clear from context, we overload $\concL\tau \eqdef \concL\tau \inclto \anyL\tau$ and
$\varL\tau \eqdef \varL\tau \inclto \anyL\tau$ for the inclusions into $\anyL\tau$.
Let $1_\vw$ be the template with only $\vw$ concrete locations, and let $?^\vw$ represent the template with only $\vw$ location variables.

Every morphism $\rho : \tau \to \hat\tau$ induces a morphism $\varL\rho \eqdef \varL\tau \xsto{\varL\tau} \anyL\tau \xsto{\anyL\rho} \anyL{\hat\tau}$ in $\worldC$
and every morphism in $\tempC$ is mono, and epis are exactly the morphisms that have a surjective $\anyL\rho$ component.
Moreover, concrete locations $\concL\tau$ cannot be mapped to variables $\varL{\hat\tau}$ by $\rho$ and so concrete locations cannot become uninstantiated.

We select specific locations by intersecting the images, \eg
$\concL{\hat\tau} \pltimes{\anyL{\hat\tau}} \varL\tau \eqdef \img{\concL{\hat\tau}} \setInt \img{\varL\rho}$, 
refers to the variables $\varL\tau$ that are instantiated by $\rho$
and similarly $\varL{\hat\tau} \pltimes{\anyL{\hat\tau}} (\anyL{\hat\tau} \ominus \anyL\rho)$ represents variables $\varL{\hat\tau}$ that are fresh w.r.t. $\rho$
($\anyL{\hat\tau} \ominus \anyL\rho$ are the new locations introduced by $\rho$).

Let $\vl$ range over concrete locations, and let $\va, \vb$ range over location variables. 
We write templates as single worlds $\vl_1 : \vC_1, \va_2 : \vC_2, \vl_3 : \vC_3$, 
to really mean the template given by the inclusion $\vl_1 : \vC_1, \vl_3 : \vC_3 \inclto \vl_1 : \vC_1, \va_2 : \vC_2, \vl_3 : \vC_3$.

\begin{example}\label{ex:tempC-overview}
Consider the following morphism in $\tempC$:
\[ 
    \vl_1 : \vC_1, \va_2 : \vC_2
    \xrightarrow{\vl_1 \mapsto \vl_1, \va_2 \mapsto \vl_2}   
    \vl_1 : \vC_1, \vl_2 : \vC_2, \va_3 : \vC_3
\]
It instantiates variable $\va_2$ from its domain, and introduces a new fresh variable $\va_3$.
\end{example}

The category $\tempC$ has initial object $?^\emptyset$, and equalizers and pullbacks such that $\concL\dash$, $\anyL\dash$ preserve these.
It is also equivalent to the arrow category $\arrC\worldC$.

\subsubsection{Local Independent Coproducts}\label{sec:lic-tempC}

For a span $\tau_1 \xsfrom{\rho_1} \tau \xsto{\rho_2} \tau_2$ its local independent coproduct
gives the least way to close the span into a commuting square, that is, it makes the least assumptions on which locations  
from $\tau_1$ and $\tau_2$ are related (at both the $\concL\dash$ and $\anyL\dash$ level).
We define its local independent coproduct $\tau_1 \liplus\tau \tau_2$ on the $\anyL\dash$ level as
\begin{center}
\begin{tikzcd}
    {\tau} \ar[r, "\rho_2"] \ar[d, "\rho_1"]
        & {\tau_2} \ar[d, "\ext^{\rho_2}\rho_1"]
    \\
    {\tau_1} \ar[r, "\ext^{\rho_1}\rho_2"]
        & {\tau_1 \liplus\tau \tau_2}
\end{tikzcd}
\hspace{10pt}
\begin{tikzcd}[sep=small]
    {\anyL\tau} 
        \ar[d, "{\anyL{\rho_2}}"' xshift=-1pt] \ar[rr, "{\anyL{\rho_1}}"] 
            && {\anyL{\tau_1}} 
                \ar[d, "{\ext^{\anyL{\rho_1}}\anyL{\rho_2}}" xshift=1pt] 
    \\
    {\anyL{\tau_2}} 
        \ar[rr, "{\ext^{\anyL{\rho_2}}\anyL{\rho_1}}"'] 
            && {\anyL\tau \oplus (\anyL{\tau_1} \ominus \anyL{\rho_1}) \oplus (\anyL{\tau_2} \ominus \anyL{\rho_2}) \defeq \anyL{\tau_1 \liplus\tau \tau_2}}
\end{tikzcd}
\end{center}
$\ext^{\anyL{\rho_1}}\anyL{\rho_2}, \ext^{\anyL{\rho_2}}\anyL{\rho_1}$ are given by the canonical isomorphism of injections:
$\anyL{\tau_\vi} \iso \anyL\tau \oplus (\anyL{\tau_\vi} \ominus \anyL{\rho_\vi})$.

The concrete locations $\concL{\tau_1 \liplus\tau \tau_2}$ are
\[ \concL{\tau_1 \liplus\tau \tau_2} \eqdef \img{\ext^{\anyL{\rho_1}}\anyL{\rho_2} \comp \concL{\tau_1}} \setInt \img{\ext^{\anyL{\rho_2}}\anyL{\rho_1} \comp \concL{\tau_2}} \]
Thus we only relate locations $\anyL\tau$ from $\anyL{\tau_1}$ and $\anyL{\tau_2}$ and
a location $\vl : \vC \in \anyL{\tau_1 \liplus\tau \tau_2}$ is concrete iff it is concrete in either $\concL{\tau_1}$ or $\concL{\tau_2}$.

\begin{example}\label{ex:tempC-lic}
\begin{tikzcd}[sep=small]
    {\va_1 : \sortNat} \ar[d, "{\va_1 \mapsto \vl_1}"' xshift=-2pt] \ar[r, hook, ]
        & {\va_1, \vl_2: \sortNat} \ar[d, "\va_1 \mapsto \vl_1" xshift=1pt]
    \\
    {\vl_1, \va_3 : \sortNat} \ar[r, hook]
        & {\vl_1, \vl_2, \va_3 : \sortNat}
    \end{tikzcd}
is a local independent coproduct square since $\va_3$ and $\vl_2$ are not conflated and $\va_3$ is not instantiated.
\end{example}

Local independent coproducts can be given a universal property using the independence structures of \cite{simpson2018category}.
A commuting square is called \emph{independent} if its $\anyL\dash$-image is a pullback, and a local independent coproduct square
is the universal independent square, that is, every other independent square factors uniquely through it.

\subsubsection{Variable Preserving Morphisms}\label{sec:vpres-morph}

We call morphisms in $\tempC$ that do not instantiate location variables from their domain as \emph{variable preserving}.
\begin{definition}
A morphism $\rho : \tau_1 \to \tau_2$ is variable preserving if it satisfies $\concL{\tau_2} \pltimes{\anyL{\tau_2}} \varL{\tau_1} = \emptyset$.
\end{definition}
\begin{example}
$\va_1 : \vC_1 \inclto \va_1 : \vC_1,\va_2 : \vC_2$ is variable preserving since $\va_1$ remains a variable in the codomain,
whereas $\va : \vC \xto{\va \mapsto \vl} \vl : \vC$ is not.
\end{example}
Variable preserving morphisms are closed under composition, and characterise exactly the regular monos in $\tempC$.

\begin{lemma}\label{lem:lic-vpres-refl}
    Local independent coproducts reflect variable preserving morphisms, that is, 
    in a span $\tau_1 \xsfrom{\rho_1} \tau \xsto{\rho_2} \tau_2$, if $\rho_2$ is variable preserving,
    then so is $\tau_1 \xsto{\ext^{\rho_1}\rho_2} \tau_1 \liplus\tau \tau_2$.
\end{lemma}

\subsubsection{Variable Reflecting Morphisms}\label{sec:vrefl-morph}

Similarly, we call a morphism that does not introduce location variables as \emph{variable reflecting}.
\begin{definition}
A $\rho : \tau_1 \to \tau_2$ is variable reflecting if it satisfies $(\anyL{\tau_2} \ominus \varL\rho) \pltimes{\anyL{\tau_2}} \varL{\tau_2} = \emptyset$.
\end{definition}
\begin{example}
    $\va : \vC \xto{\va \mapsto \vl} \vl : \vC$ is a variable reflecting whereas
    $\va_1 : \vC_1 \inclto \va_1 : \vC_1, \va_2 : \vC_1$ is not.
\end{example}

Variable reflecting morphisms are closed under composition and local independent coproducts reflect them in the same manner as \cref{lem:lic-vpres-refl}.

\subsection{Template Instantiations}\label{sec:tinstC}

We extend templates and their morphisms with instantiation data. 
Let $\tinstC$ be the category with objects being templates $\tau \in \objC\tempC$,
and morphisms $\tau_1 \xto\me \tau_2 \in \tinstC$ are pairs consisting of a template morphism 
$\tau_1 \xto{\tinstU\me} \tau_2 \in \tempC$ and instantiation data
$\mvarM_\me : \dheapF(\concL{\tau_2} \ominus \concL{\tinstU\me}, \anyL{\tau_2})$.
$\mvarM_\me$ defines a value that has access to all $\anyL{\tau_2}$ for every $\vl : \vC \in \concL{\tau_2}$ that is not hit by $\concL{\tinstU\me}$.

Composition of $\tau_1 \xto{\me_1} \tau_2 \xto{\me_2} \tau_3$ is defined by composition in $\tempC$ for the underlying injection and composition of the instantiation data 
$\mvarM_{\me_2 \comp \me_1} : \dheapF(\concL{\tau_3} \ominus (\concL{\me_2} \comp \concL{\me_1}), \anyL{\tau_3})$ is given as in $\instC$:
\[
    \vl : \vC \mapsto \begin{cases}\sem{\ctype\vC}\anyL{\me_2}(\proj_{\vl'}\mvarM_{\me_1}) & \mif \concL{\me_2}(\vl') = \vl \\
                                   \proj_\vl{\mvarM_{\me_2}} & \mif \vl \in \concL{\tau_3} \ominus \concL{\me_2}
                      \end{cases}
\]

We have a forgetful functor $\tinstU : \tinstC \to \tempC$, lifted functors $\concL\dash, \anyL\dash : \tinstC \to \worldC$,
and functor $\inclITI\dash : \instC \inclto \tinstC$ which maps to concrete templates $\vw \mapsto 1_\vw$.
We further inherit the terminology of $\tempC$ and talk about concrete locations and location variables, 
and talk about variable preserving and reflecting morphisms by noticing that $\tinstU$ preserves and reflects these properties.
The forgetful functor preserves equalizers and pullbacks.

\begin{definition}\label{def:template-heap} The template heap functor is defined as $\theapF \eqdef \homF\tinstC{?^\emptyset}\dash: \tinstC \to \setC$.
\end{definition}
We have $\theapF(\tau) \iso \dheapF(\concL\tau, \anyL\tau)$ and so a template heap $\mvarM \in \theapF(\tau)$ defines a value for every concrete location in $\concL\tau$ that has access to all locations $\anyL\tau$.

\subsubsection{Reachability}

Instantiation data introduce a notion of (un)reachable locations, refering to whether a location can be reached from a concrete location by traversing the data.
 
\begin{definition}\label{def:reach-loc}
    The reachable locations from $\vw \xto\sigma \anyL\tau$ w.r.t. a template heap $\mvarM : \theapF(\tau)$
    is the least factoring of $\sigma$ as 
    \begin{tikzcd} 
        {\vw} \ar[r, "\exists\sigma_\reach"] \ar[rr, bend right=20, "\sigma"'] & {\vw_\reach} \ar[r, "\exists\sigma_\ureach"] & {\anyL\tau} 
     \end{tikzcd}
     such that for every reachable concrete location $\vb : \vC \in \vw_\reach \pltimes{\anyL\tau} \concL\tau$, 
     the locations in $\proj_\vb \mvarM$ are included in $\vw_\reach$: 
     there exists a $\vV \in \sem{\ctype\vC}(\vw_\reach)$ s.t. $\sem{\ctype\vC}(\sigma_\ureach)(\vV) = \proj_\vb \mvarM$.
     We refer to $\vw_\reach$ as the reachable and $\anyL\tau \ominus \sigma_\ureach$ as the unreachable locations.
\end{definition}
 
\begin{example}\label{ex:reach-loc}
\renewcommand\ndsize{15}
    The reachable locations from $\vl_1 : \sortNat \inclto (\vl_1, \va_2, \vl_3 : \sortNat)$ w.r.t. template heap
    \begin{tikzpicture}[baseline={(l1.base)}]
        \node[ccell=$\vl_3$] (l3) {};
        \node[ccell=$\vl_1$] (l1) [below=0pt of l3] {};
        \node[vcell=$\va_2$] (a2) [below=0pt of l1] {};
        \draw[ptr] (l1.center) to ++\east to ++\south to (a2.east);
        \draw[ptr] (l3.center) to ++\east to ++\ssouth to ($ (l1.east) + (0,\gap pt) $);
    \end{tikzpicture}\hspace{1pt}
    are $\vl_1 : \sortNat \inclto (\vl_1, \va_2 : \sortNat) \inclto (\vl_1, \va_2, \vl_3 : \sortNat)$, since $\vl_3$ is not reachable from $\vl_1$.
\end{example}
 
\begin{lemma}\label{lem:reach-loc-dec}
    For every $\vw \xto\sigma \anyL\tau$ and $\mvarM : \theapF(\tau)$ the reachable location factoring $\sigma_\reach,\sigma_\ureach$ exists.
\end{lemma}
\begin{proofsketch} 
    World $\vw^\reach$ is given by the image $\img\sigma$ and all locations encountered when traversing $\mvarM$ starting from every concrete location that is in $\vw$ (\ie $\vw \pltimes{\anyL\tau} \concL\tau$).
\end{proofsketch}
 
\subsubsection{Connectivity}\label{sec:connectivity}

Using reachability we identify \emph{connected} morphisms in $\tinstC$ with the property that all locations in the codomain are reachable from locations in the domain.
We define connectivity formally by giving a factorisation of morphisms in $\tinstC$ that splits its reachable and unreachable components.
 
\begin{definition}[Factorisation in $\tinstC$]\label{def:con-fact}
    For a morphism $\me : \tau_1 \to \tau_2$, 
    we consider the reachable locations from $\anyL\me$ w.r.t. 
    $\mvarM_\me : \theapF(\concL{\tau_2} \ominus \concL\me \xto{\compl{\concL\me}} \concL{\tau_2} \xto{\concL{\tau_2}} \anyL{\tau_2})$.
    Let
    \begin{tikzcd}
        {\anyL{\tau_1}} 
            \ar[r, "{\anyL\me}^\reach"] \ar[rr, bend right=20, "\anyL\me"'] 
            & {\vw^\reach} 
                \ar[r, "{\anyL\me}^\ureach"] 
                & {\anyL{\tau_2}} 
    \end{tikzcd}
    be the factoring of $\me$ with the chosen subset for $\vw^\reach$ being the reachable concrete locations $\vw^\reach \pltimes{\anyL{\tau_2}} \concL{\tau_2}$.
    The property of \cref{def:reach-loc} allows us to split $\mvarM_\me$ into data for these two morphisms.

    For a morphism $\me : \tau_1 \to \tau_2$, we write $\con\me : \tau_1 \to \tau_{\con\me}$ and $\uncon\me : \tau_{\con\me} \to \tau_2$ for its factorisation given by \cref{def:con-fact}.
\end{definition}
\begin{example}
\renewcommand\ndsize{12}
The factoring of $(\va_1 : \sortNat) \xto{\va_1 \mapsto \vl_1} (\vl_1, \va_2, \vl_3 : \sortNat)$ with data
\begin{tikzpicture}[baseline={(l1.base)}]
    \node[ccell=$\vl_3$] (l3) {};
    \node[ccell=$\vl_1$] (l1) [below=0pt of l3] {};
    \node[vcell=$\va_2$] (a2) [below=0pt of l1] {};
    \draw[ptr] (l1.center) to ++\east to ++\south to (a2.east);
    \draw[ptr] (l3.center) to ++\east to ++\ssouth to ($ (l1.east) + (0,\gap pt) $);
\end{tikzpicture}
is 
\[ (\va_1 : \sortNat) \xto{\va_1 \mapsto \vl_1} (\vl_1, \va_2 : \sortNat) \inclto (\vl_1, \va_2, \vl_3 : \sortNat) \text{ with data }
\begin{tikzpicture}
    \node[ccell=$\vl_1$] (l1) [below=0pt of l3] {};
    \node[vcell=$\va_2$] (a2) [below=0pt of l1] {};
    \draw[ptr] (l1.center) to ++\east to ++\south to (a2.east);
\end{tikzpicture} \text{ and }
\begin{tikzpicture}
    \node[ccell=$\vl_3$] (l3) {};
    \node[vcell=$\vl_1$] (l1) [below=0pt of l3] {};
    \node[vcell=$\va_2$] (a2) [below=0pt of l1] {};
    \draw[ptr] (l3.center) to ++\east to ++\south to (l1.east);
\end{tikzpicture}
\]
\end{example}
 
A morphism $\me : \tau_1 \to \tau_2$ is said to be connected if $\uncon\me$ is an isomorphism
and correspondingly, it is disconnected if $\con\me$ is iso.
We have that $\con\me$ is itself connected, and $\uncon\me$ is disconnected.

Epimorphisms in $\tinstC$ are exactly the connected morphisms, 
and regular monomorphisms are exactly the disconnected morphisms. 
Disconnected morphisms are alternatively characterised as the variable preserving ones.
We use arrows $\surjto$ for epis, and $\rto$ for regular monos.

\begin{lemma}[Orthogonal Factorisation System in $\tinstC$]\label{lem:ofs-tinstC}
    The classes of epis (connected) and regular monos (disconnected) morphisms gives rise to an orthogonal factorisation system in $\tinstC$.
\end{lemma}

\subsubsection{Unification}

Dealing with concrete cells that can store location variables naturally leads us to unification that answers the question
of whether two template heaps with a common set of starting locations unify to a larger template heap that includes them both. 
In the case of $\tinstC$, template heaps are morphisms and so we phrase unification as follows.
\begin{definition}
    A span $\tau_1 \xfrom{\me_1} \tau \xto{\me_2} \tau_2$ unifies if there exists a commuting square
    \partsmash[tb]{0.5}{
    \begin{tikzcd}[ampersand replacement=\&]
        {\tau} \ar[d, "\me_1"'] \ar[r, "\me_2"] 
            \& {\tau_2} \ar[d, "\me_2'"]
        \\     
        {\tau_1} \ar[r, "\me_1'"']
            \& {\hat\tau}
    \end{tikzcd}}
    We call $\me_1', \me_2'$ a unifier of $\me_1, \me_2$.
\end{definition}
\begin{example}
    Not all spans unify, \eg:
    \renewcommand\ndsize{15}
    \begin{tikzcd}[remember picture]
        |[alias=s1]|{(\vl : \sortBool)}
            & {(\va : \sortBool)} \ar[l] \ar[r]
                & |[alias=s2]|{(\vl : \sortBool)}
    \end{tikzcd}
    \begin{tikzpicture}[overlay, remember picture] 
        \node[ccell=$\vl$] (d1) [above=-4pt of s1, xshift=35pt] {\tiny $\False$};
        \node[ccell=$\vl$] (d2) [above=-4pt of s2, xshift=-35pt] {\tiny $\True$};
    \end{tikzpicture}
\end{example}
\noindent
We consider a partial order on unifiers of the same span defined as
\[ 
    \tau_1 \xto{\me_1} \tau \xfrom{\me_2} \tau_2 \;\le\;
    \tau_1' \xto{\me_1'} \tau' \xfrom{\me_2'} \tau_2' \quad\miff\quad
\begin{tikzcd}[sep=small]
    {\tau_1} \ar[r, "\me_1"] \ar[dr, "\me_1'"', end anchor=west] 
        & {\tau} \ar[d, "\exists \vh" yshift=2pt]
            & {\tau_2} \ar[l, "\me_2"'], \ar[dl, "\me_2'", end anchor=east]
    \\
        & {\tau'}
\end{tikzcd}
\]
In general, a span might not have a least unifier as the following example shows.
\begin{example}\label{ex:single-mgu-counterexample}
    \renewcommand\ndsize{8}
    Consider the following two unifiers of the same span with the underlying $\tempC$ morphisms given by the indices.\\
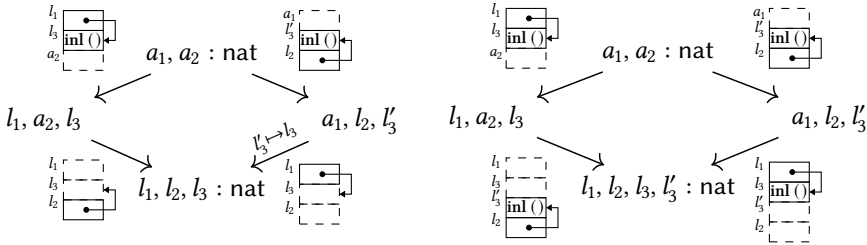
\begin{center}
\renewcommand\ndsize{15}
    \begin{tikzcd}[remember picture, sep=small]
            & {\va_1, \va_2 : \sortNat} \ar[dl] \ar[dr]
        \\ 
        |[alias=ll]|{\vl_1, \va_2, \vl_3} \ar[dr]
                && |[alias=lr]|{\va_1, \vl_2, \vl_3'} \ar[dl, sloped, "\vl_3' \mapsto \vl_3" xshift=1pt]
        \\
            & {\vl_1, \vl_2, \vl_3 : \sortNat} 
    \end{tikzcd}
    \hspace{5pt}
    \begin{tikzcd}[remember picture,sep=small]
            & {\va_1, \va_2 : \sortNat} \ar[dl] \ar[dr]
        \\ 
        |[alias=rl]|{\vl_1, \va_2, \vl_3} \ar[dr] 
                && |[alias=rr]|{\va_1, \vl_2, \vl_3'} \ar[dl] 
        \\
            & {\vl_1, \vl_2, \vl_3, \vl_3' : \sortNat}
    \end{tikzcd}
    \begin{tikzpicture}[overlay, remember picture] 

        \node[ccell=$\vl_1$] (ll1) [above=25pt of ll, xshift=15pt] {};
        \node[ccell=$\vl_3$] (ll3) [below=0pt of ll1] {\tiny $\InjL\Unit$};
        \node[vcell=$\va_2$] (ll2) [below=0pt of ll3] {};
        \draw[ptr] (ll1.center) to ++\east to ++\south to (ll3.east);
    
        \node[vcell=$\va_1$] (lr1) [above=25pt of lr, xshift=-15pt] {};
        \node[ccell=$\vl_3'$] (lr3') [below=0pt of lr1] {\tiny $\InjL\Unit$};
        \node[ccell=$\vl_2$] (lr2) [below=0pt of lr3'] {};
        \draw[ptr] (lr2.center) to ++\east to ++\north to (lr3'.east);

        \node[vcell=$\vl_1$] (ll1) [below=7.5pt of ll, xshift=15pt] {};
        \node[vcell=$\vl_3$] (ll3d) [below=0pt of ll1] {};
        \node[ccell=$\vl_2$] (ll2) [below=0pt of ll3d] {};
        \draw[ptr] (ll2.center) to ++\east to ++\north to (ll3d.east);
    
        \node[ccell=$\vl_1$] (lr1) [below=7.5pt of lr, xshift=-15pt] {};
        \node[vcell=$\vl_3$] (lr3) [below=0pt of lr1] {};
        \node[vcell=$\vl_2$] (ll2) [below=0pt of lr3] {};
        \draw[ptr] (lr1.center) to ++\east to ++\south to (lr3.east);


        \node[ccell=$\vl_1$] (rl1) [above=25pt of rl, xshift=15pt] {};
        \node[ccell=$\vl_3$] (rl3) [below=0pt of rl1] {\tiny $\InjL\Unit$};
        \node[vcell=$\va_2$] (rl2) [below=0pt of rl3] {};
        \draw[ptr] (rl1.center) to ++\east to ++\south to (rl3.east);
    
        \node[vcell=$\va_1$] (rr1) [above=25pt of rr, xshift=-15pt] {};
        \node[ccell=$\vl_3'$] (rr3') [below=0pt of rr1] {\tiny $\InjL\Unit$};
        \node[ccell=$\vl_2$] (rr2) [below=0pt of rr3'] {};
        \draw[ptr] (rr2.center) to ++\east to ++\north to (rr3'.east);

        \node[vcell=$\vl_1$] (rl1) [below=7.5pt of rl, xshift=15pt] {};
        \node[vcell=$\vl_3$] (rl3) [below=0pt of rl1] {};
        \node[ccell=$\vl_3'$] (rl3') [below=0pt of rl3] {\tiny $\InjL\Unit$};
        \node[ccell=$\vl_2$] (rl2) [below=0pt of rl3'] {};
        \draw[ptr] (rl2.center) to ++\east to ++\north to (rl3'.east);
    
        \node[ccell=$\vl_1$] (rr1) [below=7.5pt of rr, xshift=-15pt] {};
        \node[ccell=$\vl_3$] (rr3) [below=0pt of rr1] {\tiny $\InjL\Unit$};
        \node[vcell=$\vl_3'$] (rr3') [below=0pt of rr3] {};
        \node[vcell=$\vl_2$] (rr2) [below=0pt of rr3'] {};
        \draw[ptr] (rr1.center) to ++\east to ++\south to (rr3.east);
    \end{tikzpicture}
    \vspace{15pt}
\end{center}
\renewcommand\ndsize{15}
The resulting heaps
    \begin{tikzpicture}[baseline={(l1.base)}]
        \node[ccell=$\vl_1$] (l1) [] {};
        \node[ccell=$\vl_3$] (l3) [below=0pt of l1] {\tiny $\InjL\Unit$};
        \node[ccell=$\vl_2$] (l2) [below=0pt of l3] {};
        \draw[ptr] (l1.center) to ++\east to ++\south to (l3.east);
        \draw[ptr] (l2.center) to ++\east to ++\north to (l3.east);
    \end{tikzpicture}
    and
    \begin{tikzpicture}[baseline={(l1.base)}]
        \node[ccell=$\vl_1$] (l1) [] {};
        \node[ccell=$\vl_3$] (l3) [below=0pt of l1] {\tiny $\InjL\Unit$};
        \node[ccell=$\vl_3'$] (l3') [below=0pt of l3] {\tiny $\InjL\Unit$};
        \node[ccell=$\vl_2$] (l2) [below=0pt of l3'] {};
        \draw[ptr] (l1.center) to ++\east to ++\south to (l3.east);
        \draw[ptr] (l2.center) to ++\east to ++\north to (l3'.east);
    \end{tikzpicture}
    are incomparable as $\vl_3, \vl_3'$ cannot be collapsed to a one cell.
\end{example}

\paragraph{General Unifiers}
 
While least unifiers do not necessarily exist, there exist unifiers which we call \emph{general}
with the property that when taken together jointly serve as the least element.

\begin{definition}\label{def:general-unifier}
    A unifier $\tau_1 \xsto{\me_1'} \tau' \xsfrom{\me_2'} \tau_2$
    for a span $\tau_1 \xsfrom{\me_1} \tau \xsto{\me_2} \tau_2$
    is called \emph{general} if $\concL{\me_1'},\concL{\me_2'}$ are jointly surjective and $\anyL{\me_1'},\anyL{\me_2'}$ are jointly surjective.
\end{definition}

Intuitively, a general unifier unifier does not introduce new concrete locations or variables that do not come from the span,
that is, $\mvarM_{\me_1'}$ is completely determined by $\mvarM_{\me_2}$ and vice versa.

\begin{example}
    Both unifiers from \cref{ex:single-mgu-counterexample} are most general unifiers.
\end{example}

The jointly surjective property of general unifiers guarantees that every unifier uniquely factors through exactly one isomorphic class of general unifiers.
It also implies that for a given span, there can only be a finite number of isomorphic classes of its general unifiers.

\paragraph{Unification and Local Independence Coproducts}\label{def:ligu}

Local independent coproducts in $\tempC$ (see \cref{sec:lic-tempC}), when equipped with instantiation data are themselves general unifiers.   
\begin{definition}[Local Independent General Unifiers]
    A unifier is called a local independent general unifier if its $\tinstU$ image is a local independent coproduct square.
\end{definition}

Consider a span $\tau_1 \xfrom{\me_1} \tau \xrto{\me_2} \tau_2$, such that $\me_2$ is variable preserving.
Since $\me_2$ does not instantiate any variables, the data of $\mvarM_{\me_2}$ does not have to agree with any cells in $\mvarM_{\me_1}$   
for the morphisms to commute.
As a result, the local independent general unifier exists and relates locations in $\tau_1$ and $\tau_2$ only from ones that come from $\tau$.

\begin{example}\label{ex:ligu}
    \renewcommand\ndsize{15}
    The square
    \begin{tikzcd}[remember picture, sep=small]
    |[alias=a1]|{\va_1 : \sortNat} \ar[r] \ar[d, hook]
        & {\vl_1, \va_3 : \sortNat} \ar[d, hook]
    \\
    {\va_1, \vl_2: \sortNat} \ar[r]
        & |[alias=lic]|{\vl_1, \vl_2, \va_3 : \sortNat}
    \end{tikzcd}
    from \cref{ex:tempC-lic} equipped with data is a local independent general unifier.
    \begin{tikzpicture}[overlay, remember picture] 
        \node[vcell=$\va_3$] (a3) [above=5pt of a1, xshift=30pt] {};
        \node[ccell=$\vl_1$] (l1') [below=0pt of a3] {\tiny $\InjR\bullet$};
        \node[ccell=$\vl_2$] (l2) [below=0pt of a1, xshift=-12pt] {\tiny $\InjL\Unit$};
        \node[ccell=$\vl_2$] (l2') [above=0pt of lic, xshift=20pt] {\tiny $\InjL\Unit$};
        \node[vcell=$\va_3$] (a3') [below=-3pt of lic, xshift=-35pt] {};
        \node[ccell=$\vl_1$] (l1'') [below=0pt of a3'] {\tiny $\InjR\bullet$};
        \draw[ptr] (l1'.east) to ++\east to ++\north to (a3.east);
        \draw[ptr] (l1''.east) to ++\east to ++\north to (a3'.east);
    \end{tikzpicture}
\end{example}

\subsection{Families of Template Instantiations}\label{sec:famC-tinstC}

The last building block required to describe the finitary monad is a means to talk about collections of template heaps related by a common set of public locations.
We model these collections using the $\famC\tinstC$ construction, 
representing finite families of templates and their instantiations:
\begin{itemize}
    \item Objects: A finite set $\vI$, and an $\vI$-indexed set of templates $\tlistr{\tau_\vi \in \objC\tinstC}{\vi \in \vI}$.
    \item Morphisms: A $\tlistr{\tau_\vi}{\vi \in \vI} \xto\phi \tlistr{\tau_\vj}{\vj \in \vJ}$ consists of a mapping $\indF\phi : \vJ \to \vI$ and a $\vJ$-indexed set of morphisms $\tlistr{\phi_\vj : \tau_{\indF\phi(\vj)} \to \tau_\vj}{\vj \in \vJ}$
\end{itemize}
Notice that $\famC\tinstC$ is the free finite product completion of $\tinstC$.
We have a forgetful functor $\indF\dash : \famC\tinstC \to \opC\finC$ that maps into the indexing sets.
Let $\Sigma$ range over objects in $\famC\tinstC$ and we use list notation to write objects in $\famC\tinstC$.

\begin{example}\label{ex:famC-tinstC-example}
\renewcommand\ndsize{18}
\[\begin{array}{lr}
    \meZ \eqdef
    \begin{tikzcd}[remember picture]
        |[alias=ndl]|{(\va : \sortNat)} \ar[r, "\va \mapsto \vl"'] & {(\vl : \sortNat)}
    \end{tikzcd}
    &\qquad 
    \meS \eqdef
    \begin{tikzcd}[remember picture]
        |[alias=ndr]|{(\va : \sortNat)} \ar[r, "\va \mapsto \vl_1"'] & {(\vl_1 : \sortNat, \va_2 : \sortNat)}
    \end{tikzcd}
\end{array}\]
\begin{tikzpicture}[overlay, remember picture] 
    \node[ccell=$\vl$] (l) [above=-4pt of ndl, xshift=32pt] {\tiny $\,\InjL\Unit\,$};
    \node[ccell=$\vl_1$] (l2) [above=-4pt of ndr, xshift=32pt] {};
    \node[vcell=$\va_2$] (a2) [above=0pt of l2] {};
    \draw[ptr] (l2.center) to ++\east to ++\north to (a2.east);
\end{tikzpicture}
$\meZ,\meS$ give us a morphism $[\va : \sortNat] \xto{\pair{\meZ}{\meS}} [(\vl : \sortNat); (\vl_1 : \sortNat, \va_2 : \sortNat)] \in \famC\tinstC$.
\end{example}

\begin{lemma}\label{lem:pushouts-fam-tinstC}
$\famC\tinstC$ has pushouts given by the general unifiers of $\tinstC$.
\end{lemma}
\begin{proofsketch}
    The pushout of
    \begin{tikzcd}[sep=small]
        {\Sigma} \ar[r, "\psi"] \ar[d, "\phi" xshift=2pt]
            & {\Sigma_2} \ar[d, "\ext^\psi\phi" xshift=2pt]
        \\
        {\Sigma_1} \ar[r, "\ext^\phi\psi"']
            & {\Sigma_1 \liplus\Sigma \Sigma_2}
    \end{tikzcd}
    selects the (representative) general unifiers for every pair of matching indices 
    $\{ \pair\vi\vj \in \indF{\Sigma_1} \times \indF{\Sigma_2} : \indF\phi(\vi) = \indF\psi(\vj) \}$. 
    That is, for every matching pair $\vi,\vj$, we have a finite number of representative general unifiers $\Sigma_{\GU{\phi_\vi}{\psi_\vj}}$ with morphisms 
    $\tau_{\indF\phi(\vi)} \xsto{\ext^{\phi_\vi}\psi_\vj} \Sigma_{\GU{\phi_\vi}{\psi_\vj}} \xsfrom{\ext^{\psi_\vj}\phi_\vi} \tau_{\indF\psi(\vj)}$ in $\famC\tinstC$.
    The pushout is then given by the pairing of all these morphisms
    \[ \Sigma_1 \liplus\Sigma \Sigma_2 \eqdef 
        \Prod{\pair\vi\vj \in \indF{\Sigma_1} \times \indF{\Sigma_2} \,:\, \indF\phi(\vi) = \indF\psi(\vj)} \Sigma_{\GU{\phi_\vi}{\psi_\vj}} \]
    and $\ext^\phi\psi \eqdef \langle \hdots, \vi,\vj \mapsto \ext^{\phi_\vi}\psi_\vj, \hdots \rangle$, 
        $\ext^\psi\phi \eqdef \langle \hdots, \vi,\vj \mapsto \ext^{\phi_\vj}\psi_\vi, \hdots \rangle$.
\end{proofsketch}

We further characterise collections in $\famC\tinstC$ according to how the individual morphisms in the collections interact.
\begin{definition}
    A morphism $\phi : \Sigma \to \Sigma'$ is said to be \emph{connected} if every $\phi_{\vj \in \indF{\Sigma'}}$ is connected.
\end{definition}
\begin{example}
    The morphism from \cref{ex:famC-tinstC-example} is connected as both $\meZ, \meS$ do not introduce unreachable locations.
\end{example}

\begin{definition}
    A morphism $\phi : \Sigma \to \Sigma'$ is said to be \emph{non-overlapping} if 
    every distinct pair of morphisms with the same root, 
    $\tau_\vi \xfrom{\phi_\vi} \tau_\vk \xto{\phi_\vj} \tau_\vj$, 
    do not unify.
    That is, for every $\vi, \vj \in \indF{\Sigma'}$ s.t. $\vi \neq \vj$:
    \[ 
        \indF\phi(\vi) = \indF\phi(\vj) \implies 
            \AllTwo{\tau_\vi \xto{\me_\vi} \hat\tau}{\tau_\vj \xto{\me_\vj} \hat\tau} 
        \me_\vi \comp \phi_\vi \neq \me_\vj \comp \phi_\vj
    \]
\end{definition}

\def\meOne{\me_{1}}
\begin{example}
\renewcommand\ndsize{17}
    The morphism from \cref{ex:famC-tinstC-example} is non-overlapping since any extension of $\meZ$ will have $\InjL\Unit$ at $\vl$ whereas $\meS$ will have $\InjR{\va_2}$.
    On the other hand, adding
\[ \meOne \eqdef
\begin{array}{c}
    \\
    \begin{tikzcd}[remember picture]
        |[alias=ndtt]|{(\va : \sortNat)} \ar[r, "\va \mapsto \vl_1"'] & {(\vl_1 : \sortNat, \vl_2 : \sortNat)}
    \end{tikzcd}
\end{array} \] 
to the collection breaks the non-overlapping property.
\begin{tikzpicture}[overlay, remember picture] 
    \node[ccell=$\vl_1$] (l1) [above=-4pt of ndtt, xshift=32pt] {};
    \node[ccell=$\vl_2$] (l2) [above=0pt of l1] {\tiny $\,\InjL\Unit\,$};
    \draw[ptr] (l1.center) to ++\east to ++\north to (l2.east);
\end{tikzpicture}
\end{example}

Epimorphisms in $\famC\tinstC$ are exactly the connected and non-overlapping morphisms.

\begin{definition}
We say that an epi $\phi : \Sigma \to \Sigma'$ is \emph{complete} if for every morphism $\me : \Sigma \to 1_\vw$ 
there exists a $\tilde\me : \Sigma' \to 1_\vw$ such that
\begin{tikzcd}[sep=small]
    {\Sigma} \ar[rr, "\me"] \ar[dr, "\phi"']
            && {1_\vw}
    \\
        & {\Sigma'} \ar[ur, "\exists \tilde\me"']
\end{tikzcd}
\end{definition}
The epi property of $\phi$ guarantees that the induced morphisms are unique, and complete morphisms are closed under composition and are stable under pushouts.

\begin{example}
    The collection $\pair\meZ\meS$ is complete: a morphism $\va : \sortNat \xto\me \inclITI\vw$ is a heap with $\vw$-many concrete locations and a public location given by $\anyL\me(\va)$. 
    Since $\pair\meZ\meS$ covers both possible inhabitants of $\anyL\me(\va)$, every such $\me$ can be factored through either $\meZ$ or $\meS$.
\end{example}


\section{Finitary FGLS Monad}

We now connect the combinatorial machinery developed in \cref{sec:theory-templates} with the goal of describing the finitary FGLS computations.
The central idea is that a finitary computation can be expressed using finitely-many stateful values with variables, each assigned to one stencil in a complete cover.
We first formally define \emph{stencils} and \emph{covers}~(\cref{sec:covers}),
we then define collections of stateful values using \emph{resultlets}~(\cref{sec:evaluation}),
and link covers and resultlets to local state computations using \emph{realisations} in (\cref{sec:evaluation}).
We finally present the finitary monad in (\cref{sec:fin-monad}).

\subsection{Covers and Matching}\label{sec:covers}

We model stencils (\cref{sec:overview-fin-fgls}), that is collections of heaps with a common set of public locations as morphisms in $\tinstC$.
\begin{definition}[Stencil Functor]
    \begin{align*}
    \stencilF\vw & = \homF\tinstC{?^\vw}\dash \iso \homF\worldC\vw{\anyL\dash} \times \theapF : \tinstC \to \setC
    \end{align*}
    with the obvious functorial action.
    A stencil $?^\vw \to \tau$ that does not have unreachable locations is said to be connected, \ie it is a connected morphism (\cref{sec:connectivity}).
\end{definition}

We call finite collections of stencils with a common set of public locations as \emph{covers}.
\begin{definition}[Cover]
    For a finite family of templates $\Sigma \in \famC\tinstC$, and world $\vw \in \objC\worldC$,
    a \emph{cover} $\cov$ over $\pair\Sigma\vw$ is an indexed family of stencils $\Prod{\tau_\vi \in \Sigma} \stencilF\vw{\tau_\vi}$,
    that is, it is a morphism $\cov :\, ?^\vw \to \Sigma \in \famC\tinstC$.
\end{definition}
\begin{example}
    $\pair\meZ\meS$ from \cref{ex:famC-tinstC-example} is a cover over $\pair{(\va : \sortNat)}{[(\vl : \sortNat); (\vl_1 : \sortNat, \va_2 : \sortNat)]}$. 
\end{example}

We inherit the terminology from \cref{sec:famC-tinstC}, and talk about connected, non-overlapping, complete covers,
and we further describe three operations on covers:
(i) hiding some public locations, (ii) rebasing over a different world and (iii) intersecting two covers.  

\begin{definition}[Hiding Locations in a Cover]
For a renaming $\sigma : \vw \injto \vw'$, and a cover $\cov$ over $\pair{\vw'}\Sigma$, we have a cover over $\pair\vw\Sigma$
by precomposition
\[ \hideSt_\sigma(\cov) \eqdef\; ?^\vw \xto{?^\sigma} ?^{\vw'} \xto \cov \Sigma \]
\end{definition}
Hiding does not preserve connectivity, non-overlapping, or completeness in general.

\begin{definition}[Rebasing a cover]\label{def:cov-rebase}
For a $\sigma : \vw \injto \vw'$ and cover $\cov$ over $\pair\vw\Sigma$,
we rebase $\cov$ to a cover $\cov^\sigma$ over $\pair{\vw'}{\Sigma^\sigma}$ 
using the pushout 
\begin{tikzcd}
    {?^\vw}
        \ar[d, "\cov"'] \ar[r, "?^\sigma"]
        & {?^{\vw'}}
            \ar[d, "\cov^\sigma"]
    \\
    {\Sigma} \ar[r, "\rebase^\sigma"' yshift=-2pt]
        & {\Sigma^\sigma} \posquare{ul}
\end{tikzcd}
in $\famC\tinstC$ (see \cref{lem:pushouts-fam-tinstC}).
\end{definition}
\noindent
Cover rebasing preserves connectivity, epimorphic covers, and completeness.

\begin{example}
\renewcommand\ndsize{17}
Consider a singleton cover $(\va : \sortNatTwo) \xto{\va \mapsto \vl_3} [(\vl_1, \vl_2 : \sortNat, \vl_3 : \sortNatTwo)]$ with data
\begin{tikzpicture}[baseline=(l3.base)]
    \node[ccell=$\vl_1$] (l1) {\tiny $\,\InjL\Unit\,$};
    \node[ccell=$\vl_3$] (l3) [below=0pt of l1] {};
    \node[ccell=$\vl_2$] (l2) [below=0pt of l3] {\tiny $\,\InjL\Unit\,$};
    \draw[ptr] (l3.center) to ++\east to ++\north to (l1.east);
    \draw[ptr] (l3.center) to ++\east to ++\south to (l2.east);
\end{tikzpicture}
and we rebase it over $(\va : \sortNatTwo) \inclto (\va : \sortNatTwo, \vb : \sortNat)$.
The resulting cover has three stencils
\[(\va_1, \va_2 : \sortNat) \to [(\vl_1, \vl_2 : \sortNat, \vl_3 : \sortNatTwo), (\vl_1, \vl_2 : \sortNat, \vl_3 : \sortNatTwo), (\vl_1, \vl_2, \va_4 : \sortNat, \vl_3 : \sortNatTwo)] \]
\begin{center}
\hspace{25pt}
\begin{tikzpicture}[baseline=(l1.base)]
        \node[ccell=$\vl_1$] (l1) {\tiny $\,\InjL\Unit\,$};
        \node[ccell=$\vl_3$] (l3) [below=0pt of l1] {};
        \node[ccell=$\vl_2$] (l2) [below=0pt of l3] {\tiny $\,\InjL\Unit\,$};
        \node [below=5pt of l2] () {$\va \mapsto \vl_1, \vb \mapsto \vl_1$};
        \draw[ptr] (l3.center) to ++\east to ++\north to (l1.east);
        \draw[ptr] (l3.center) to ++\east to ++\south to (l2.east);
\end{tikzpicture}
\hspace{10pt}
\begin{tikzpicture}[baseline=(l1.base)]
        \node[ccell=$\vl_1$] (l1) {\tiny $\,\InjL\Unit\,$};
        \node[ccell=$\vl_3$] (l3) [below=0pt of l1] {};
        \node[ccell=$\vl_2$] (l2) [below=0pt of l3] {\tiny $\,\InjL\Unit\,$};
        \node [below=5pt of l2] () {$\va \mapsto \vl_1, \vb \mapsto \vl_2$};
        \draw[ptr] (l3.center) to ++\east to ++\north to (l1.east);
        \draw[ptr] (l3.center) to ++\east to ++\south to (l2.east);
\end{tikzpicture}
\hspace{10pt}
\begin{tikzpicture}[baseline=(l3.base)]
        \node[ccell=$\vl_1$] (l1) {\tiny $\,\InjL\Unit\,$};
        \node[ccell=$\vl_3$] (l3) [below=0pt of l1] {};
        \node[ccell=$\vl_2$] (l2) [below=0pt of l3] {\tiny $\,\InjL\Unit\,$};
        \node[vcell=$\va_4$] (a4) [below=0pt of l2] {};
        \node [below=5pt of a4] () {$\va \mapsto \vl_1, \vb \mapsto \va_4$};
        \draw[ptr] (l3.center) to ++\east to ++\north to (l1.east);
        \draw[ptr] (l3.center) to ++\east to ++\south to (l2.east);
\end{tikzpicture}
\end{center}
\end{example}

\begin{definition}[Cover Intersection]
Given a cover $\cov$ over $\pair\vw\Sigma$ and a cover $\cov'$ over $\pair\vw{\Sigma'}$ 
we define their intersection $\cov \pltimes\vw \cov'$ as the biggest cover that is included in both of them.
It is given by the pushout
\begin{tikzcd}[sep=small]
    {?^\vw}
        \ar[dd, "\cov"'] \ar[rrr, "\cov'"] \ar[ddrrr, "\cov \pltimes\vw \cov'" xshift=-3pt]
        &&& {\Sigma'}
            \ar[dd, "\ext^{\cov'}\cov"] 
    \\\\
    {\Sigma}
        \ar[rrr, "\ext^\cov\cov'"']
        &&& {\Sigma \pltimes{?^\vw} \Sigma'}
\end{tikzcd}
\end{definition}

Finally, we use the notion of covers to describe sorts that can be finitely enumerated.
\begin{definition}[Sort Enumeration]\label{def:sort-enumeration}
    A cover $?^{\va : \vC} \xto\cov \Sigma$ is called an \emph{enumeration} of $\vC$ if 
    it is complete and every $\cov_{\vi \in \indF\Sigma} : \va : \vC \to \tau_\vi$ instantiates $\va$, \ie $(\va : \vC) \pltimes{\anyL{\tau_\vi}} \varL{\tau_\vi} = \emptyset$.
\end{definition}
We call a sort \emph{finite} if it has an enumeration and every sort that has a full ground type interpretation as in \cref{sec:pl-lamref} is finite.
\begin{example}
\renewcommand\ndsize{16}
    The only enumeration of $\ctype\sortZero \eqdef \tyZero$ is the empty cover $(\va : \sortZero) \to \emptyset$,
    for $\ctype\sortUnit \eqdef \tyUnit$ is the singleton cover $(\va : \sortUnit) \xto{\va \mapsto \vl} [(\vl : \sortUnit)]$ with unit data.
    For $\sortLBool$ an enumeration has stencils
    $(\va : \sortLBool) \to [(\vl : \sortBool), (\vl_1 : \sortLBool, \va_2 : \sortBool), (\vl_1 : \sortLBool, \va_2 : \sortBool, \va_3 : \sortLBool)]$
    with data given by
\begin{tikzpicture}[baseline=(l.base)]
        \node[ccell=$\vl$] (l) {\tiny $\,\InjL\Unit\,$};
        \node [below=5pt of l] () {$\va \mapsto \vl$};
\end{tikzpicture},
\begin{tikzpicture}[baseline=(l1.base)]
    \node[vcell=$\va_2$] (a2) [] {};
    \node[ccell=$\vl_1$] (l1) [below=0pt of a2] {\tiny $\pair\bullet\bullet$};
    \node [below=5pt of l1] () {$\va \mapsto \vl_1$};
    \draw[cptr] ($ (l1.center) + (-3pt,-1pt) $) to ++\wwest to ++\north to ++\north to ++\eeast to ++(0pt,-4pt);
    \draw[cptr] ($ (l1.center) + (2pt,0pt) $) to ++\east to ++\south to ++\west to ++(0pt,4pt);
\end{tikzpicture}
and
\begin{tikzpicture}[baseline=(a3.base)]
    \node[vcell=$\va_2$] (a2) [] {};
    \node[ccell=$\vl_1$] (l1) [below=0pt of a2] {\tiny $\pair\bullet\bullet$};
    \node[vcell=$\va_3$] (a3) [below=0pt of l1] {};
    \node [below=5pt of a3] () {$\va \mapsto \vl_1$};
    \draw[cptr] ($ (l1.center) + (-3pt,-1pt) $) to ++\wwest to ++\north to ++\north to ++\eeast to ++(0pt,-4pt);
    \draw[cptr] ($ (l1.center) + (2pt,0pt) $) to ++\east to ++\south to ++(-6pt,0pt);
\end{tikzpicture}
\end{example}

\subsection{Evaluation}\label{sec:evaluation}

We call morphisms that are both variable preserving and reflecting (see \cref{sec:vpres-morph,sec:vrefl-morph}) as \emph{allocation morphisms}
to model that allocation only introduces new concrete locations. 
\begin{definition}[Allocation Morphisms]
    A morphism $\rho : \tau_1 \to \tau_2 \in \tempC$ is an allocation morphism if it is both variable preserving and reflecting, 
    and we denote it by $\ato$ arrow.
\end{definition}
For a $\rho : \tau_1 \ato \tau_2$, its variables $\varL{\tau_2}$ are exactly $\varL{\tau_1}$, and so
in particular, a morphism with a concrete template as domain $1_\vw \ato \tau'$ implies that its codomain is also concrete: $\tau' = 1_{\concL{\tau'}}$.

Allocation morphisms are reflected by local independent coproducts (as in \cref{lem:lic-vpres-refl}), 
and further, identities are allocation morphisms and are closed under composition.
Therefore, consider $\allocC\tempC \subseteq \tempC$, the wide subcategory of $\tempC$ that has objects $\objC\tempC$ and only allocation morphisms between them.
Let $\sliceC\tau{\allocC\tempC}$ be the set of allocation morphisms with domain $\tau$
\[ \tau \xato\rho \hat\tau \in \sliceC\tau{\allocC\tempC} \eqdef \Coprod{\hat\tau \in \objC\tempC} \homF{\allocC\tempC}\tau{\hat\tau}  \]

\paragraph{Stateful Values}

Recall that stateful values (\cref{sec:overview-fgls}) represent return values that live in the context of a heap.
They are modelled using the monad $\hideMon : \pinstC \to \pinstC$ given using a coend 
\[ \hideMon(\vA)(\vw) = \Coend{\vw \injto \vw' \in \commaC\vw\instU} \vA(\vw') \iso \Coprod{\vw \xinjto\sigma \vw' \in \commaC\vw\instU} {\vA(\vw')}_{\big/\sim} \]
$\commaC\vw\instU$ is the comma category over world $\vw \in \objC\worldC$ and $\tinstU : \instC \to \worldC$. 
The quotient $/\sim$ ensures that only observable allocations are relevant which is what models garbage collection.
In the setting of FGLS, we take $\vA = \vX \odot \heapF$ for some $\vX \in \pworldC$.
We refer the reader to \cref{sec:overview-fgls} for some examples.

\paragraph{Resultlets}
Recall that collections of heaps with a common set of public locations are modelled using stencils that have location variables.
We similarly generalise stateful values with location variables to describe collections of stateful values with a common heap footprint.
We refer to these collections of stateful values as \emph{resultlets}.
\begin{definition}
    For a $\vX : \worldC \to \setC$ and template $\tau$, a resultlet $\vr$ over $\pair\tau\vX$ is a
    \[ \vr \in \typeRes\vX\tau \eqdef \Coprod{\tau \xato\rho \tau' \in \objC{\sliceC\tau{\allocC\tempC}}} \vX\anyL{\tau'} \times \theapF(\tau') \]
The data of a resultlet are
\begin{itemize}
    \item An allocation morphism $\tau \xato\rho \tau'$ which allocates $\concL{\tau'} \ominus \concL\rho$ concrete locations.
    \item An $\vX$-valued return value $\vx : \vX\anyL{\tau'} \iso \vX(\concL{\tau'} \oplus \varL\tau)$.
    \item A template heap $\mvarM : \theapF(\tau')$, that defines locations $\concL{\tau'}$ that have access to both $\concL{\tau'}, \varL\tau$.
\end{itemize}
\end{definition}

When $\tau$ is a concrete world $1_\vw$, a resultlet has exactly the same data as a stateful value in the original hiding monad. 
Thus resultlets generalise stateful values from concrete worlds to templates with variables.
Specifically, we have a map $\quot^\vX_\vw : \typeRes\vX(1_\vw) \to \hideMon(\vX\tinstU \times \heapF)(\vw)$ for every world $\vw$:
\[ \inj_{1_\vw \xto\rho \tau'} \pair\vx\mvarM \mapsto \eqclass{\vw \xto{\concL\rho} \concL{\tau'}}\pair\vx\mvarM \]

Accordingly, we lift the actions on morphisms of the monad $\hideMon$ to resultlets.
This lifted action is not itself functorial, since it does not account for the quotienting built into $\hideMon$.
Nevertheless, this is harmless: functoriality is only required at concrete worlds, and it is recovered by the map $\quot$.

\begin{definition}[Covariant action on morphisms in $\pworldC$]
    A $\alpha : \vX \To \vY$, and $\tau \in \objC\tempC$ induces $\typeRes(\alpha)_\tau : \typeRes\vX\tau \to \typeRes\vY\tau$ defined as
    \[ \inj_{\tau \xato{\rho'} \tau'}\pair\vx\mvarM \mapsto \inj_{\rho'}\pair{\alpha_{\anyL{\tau'}}(\vx)}\mvarM \]
    It changes a resultlet over $\pair\tau\vX$ to one over $\pair\tau\vY$.
\end{definition}

\begin{definition}[Covariant action on morphisms in $\tinstC$]
For a $\me : \tau_1 \to \tau_2 \in \tinstC$ and a resultlet $\inj_{\tau_1 \xsato{\rho'} \tau'} \pair{\vx'}{\mvarM'}$,
we first lift $\rho'$ to a morphism $\me'$ in $\tinstC$ by selecting the allocation data: $\mvarM_{\me'} \eqdef \dheapF(\concL{\rho'}, 1)(\mvarM')$.
We then consider the local independent general unifier
\begin{tikzcd}[sep=small]
    {\tau_1} 
        \ar[d, "\me"' xshift=-2pt] \ar[-alloctip, rr, "\me'"] 
        && {\tau'} 
                \ar[d, "\ext^{\me'}\me" xshift=2pt] 
    \\
    {\tau_2} 
        \ar[-alloctip, rr, "\ext^\me\me'"'] 
        && {\tau_2 \liplus{\tau_1} \tau'}
\end{tikzcd} which exists because $\me'$ is variable preserving.
Define the covariant action $\typeRes\vA\me : \typeRes\vA{\tau_1} \to \typeRes\vA{\tau_2}$ as
$\inj_{\rho'}\pair\vx\mvarM \mapsto \inj_{\tinstU(\ext^\me{\me'})} \pair{\vX\anyL{\ext^{\me'}\me}(\vx)}{\theapF(\ext^{\me'}\me)(\mvarM)}$
\end{definition}

\begin{definition}[Contravariant action on morphisms in $\allocC\tempC$]
    Every allocation morphism $\tau_2 \xato{\rho'} \tau_1$ in $\tempC$ induces $\hideRes\vX{\rho'} : \typeRes\vX\tau_1 \to \typeRes\vX\tau_2$
    given by $\inj_\rho \pair\vx\mvarM \mapsto \inj_{\rho \comp \rho'} \pair\vx\mvarM$.
    We call this action \emph{hiding} as it forgets some of the public locations.
\end{definition}



We call a finite collection of resultlets as an \emph{evaluation}.
\begin{definition}[Evaluation]
    For a $\vX : \worldC \to \setC$, and $\Sigma \in \famC\tinstC$, an \emph{evaluation} is an
    $\Sigma$-indexed family of resultlets, $\ev \in \Prod{\tau_\vi \in \Sigma} \typeRes\vX{\tau_\vi}$.
\end{definition}

In summary, resultlets play the role of compatible stateful values, and an evaluation is a finite family of resultlets, one for each stencil of a cover. 
The pair consisting of a cover and an evaluation can therefore be understood as a finite decision tree: 
the cover determines which case of the input heap we are in, and the evaluation specifies the corresponding output behaviour.
Formally, we link covers and evaluations to local state computations using the notion of realisation.

\begin{definition}[Realisation]
A complete cover $\cov$ over $\pair\vw\Sigma$ and an $\vX$-valued evaluation $\ev$ over $\Sigma$, 
realises $\phi \in \monFLS\vX\vw$, notation $\cov,\ev \Vdash \phi$,
if for every $\sigma : \vw \injto \vw', \mvarM : \heapF(\vw')$, 
\[ (\proj_\sigma \phi)(\mvarM) = \vq^\vX_{\vw'}(\typeRes\vX(\tilde\me_\star)(\proj_{\indF{\tilde\me}(\star)} \ev)) \]  
where morphism $\tilde\me : \Sigma \to 1_{\vw'}$ is given by completeness: 
\begin{tikzcd}
    {?^{\vw}} \ar[r, "\cov"] \ar[rr, bend right=20, "\pair\sigma\mvarM"'] & {\Sigma} \ar[r, dashed, "\exists\tilde\me"] & {1_{\vw'}}
\end{tikzcd}
\end{definition}

\subsection{The Monad}\label{sec:fin-monad}

We are now in a position to describe the finitary fragment of the FGLS monad.

\begin{definition}[Finitary FGLS Monad]
\[\begin{array}{r@{\;}l@{\,}l@{\,}}
    \monFLSFin         :& \multicolumn{2}{l}{\pworldC \to \pworldC} \\
    \monFLSFin\vX\vw   \eqdef& \{ \phi \in \monFLS\vX\vw :& 
            \exists \Sigma \in \famC\tinstC, \textit{complete $\cov$ over $\pair\vw\Sigma$, $\vX$-valued $\ev$ over $\Sigma$ s.t. $\cov,\ev \Vdash \phi$} \}
\end{array}\]
\end{definition}

A finitary $\phi$ is one that can be realised by a finite family of templates, a complete cover and evaluation.
The existential given here is mere existence, so the actual cover and evaluation needed to realise a finitary $\phi$ is not observable, nor is it unique.

The following two theorems give the functorial and monad structure of $\monFLSFin$ by lifting it through the actions on $\monFLS$. 
One can easily show that monad $\monFLSFin$ is strong by a similar lifting of the strength for $\monFLS$, and we therefore have a strong monad morphism $\iota : \monFLSFin \To \monFLS$ given by the inclusion.
\Cref{thm:fin-fgls-monad-finitariness} lastly shows that $\monFLSFin$ is indeed finitary.

\begin{theorem}[Functor Structure]
The action of $\monFLS$ on morphisms $\hat\sigma : \vw \to \hat\vw \in \worldC$ and $\alpha : \vX \To \vY \in \pworldC$
applied to finitary computations preserve finitarity.
\begin{center}
\begin{tikzcd}
    {\monFLSFin\vX\vw}
        \ar[rr, dashed, "\exists\monFLSFin\alpha_\vw"] \ar[d, hook]
            && {\monFLSFin\vY\vw} \ar[d, hook]
    \\
    {\monFLS\vX\vw}
        \ar[rr, "\monFLS\alpha_\vw"] 
            && {\monFLS\vY\vw}
\end{tikzcd}
\hspace{10pt}
\begin{tikzcd}
    {\monFLSFin\vX\vw}
    \ar[rr, dashed, "\exists\monFLSFin\vX{\hat\sigma}"] \ar[d, hook]
            && {\monFLSFin\vX\hat\vw} \ar[d, hook]
    \\
    {\monFLS\vX\vw}
        \ar[rr, "\monFLS\vX{\hat\sigma}"] 
            && {\monFLS\vX\hat\vw}
\end{tikzcd}
\end{center}
\end{theorem}
\begin{proofsketch}
    Consider a finitary $\phi \in \monFLSFin\vX\vw$, with $\Sigma, \cov, \ev$.

    $\monFLS\alpha_\vw\phi$ is given at $\sigma : \vw \injto \vw', \vh : \heapF(\vw')$ as
    $\hideMon(\alpha \odot \heapF)_{\vw'}((\proj_\sigma \phi)(\vh))$.
    Then to show that $\monFLS\alpha_\vw\phi$ is finitary take the same cover $\cov$,
    and define a new evaluation $\ev'$ at $\tau_\vi \in \Sigma$ as
    $\proj_{\tau_\vi} \ev' \eqdef \typeRes(\alpha)_{\tau_\vi}(\proj_{\tau_\vi} \ev)$.

    $\monFLS\vX{\hat\sigma}$ is given at $\sigma' : \hat\vw \injto \vw', \vh : \heapF(\vw')$ as
    $(\proj_{\sigma' \comp \hat\sigma} \phi)(\vh)$.
    To show that $\monFLS\vX{\hat\sigma}(\phi)$ is finitary, we first rebase $\cov$ using \cref{def:cov-rebase}
    which gives a complete cover $\cov^\sigma$ over $\pair{\hat\vw}{\Sigma^\sigma}$.
    And define evaluation $\ev^\sigma$ at $\vk \in \indF{\Sigma^\sigma}$ as $\typeRes(\vX)(\rebase_\vk)(\proj_{\indF\rebase(\vk)}\ev)$.
\end{proofsketch}

\begin{theorem}[Monad Structure]
    The unit of $\monFLS$ gives a finitary computation and the join of $\monFLS$ applied to finitary computations yields a finitary one.
\begin{center}
\begin{tikzcd}
    {\vX\vw} \ar[r, dashed, "\exists\eta^\fin_\vw"] \ar[dr, "\eta_\vw"']
        & {\monFLSFin\vX\vw} \ar[d, hook]
    \\
        & {\monFLS\vX\vw}
\end{tikzcd}
\hspace{10pt}
\begin{tikzcd}[sep=small]
    {\monFLSFin\monFLSFin\vX\vw}
        \ar[rr, dashed, "\exists \mu^\fin_\vw"] \ar[d, hook]
            && {\monFLSFin\vX\vw} \ar[dd, hook]
    \\
    {\monFLS\monFLSFin\vX\vw}
        \ar[d, "\monFLS(\inclto)_\vw"' xshift=-2pt]
    \\
    {\monFLS\monFLS\vX\vw}
        \ar[rr, "\mu_\vw"] 
            && {\monFLS\vX\vw}
\end{tikzcd}
\end{center}
\end{theorem}
\begin{proofsketch}
    The unit $\eta_\vw(\vx) : \monFLSFin\vX\vw$ is given at $\sigma : \vw \to \vw', \vh : \heapF(\vw')$ as $\quot_{1_{\vw'}}\pair{\vX(\sigma)(\vx)}\vh$.
    To show that $\eta_\vw(\vx)$ is finitary, consider the identity cover $?^{\vw} \xTo{1} ?^{\vw}$. 
    It is easy to see that it is complete.
    Define the evaluation at template $?^{\vw}$ as resultlet $\inj{?^\vw \xTo{1} ?^\vw}\pair\vx\emptyset$.

    Multiplication $\mu_\vw(\monFLS(\inclto)_\vw\phi)$ involves a number of steps.
    From finitariness of $\phi$, we have a complete cover $\cov$ over $\pair\vw\Sigma$ and $(\monFLSFin\vX)$-valued evaluation $\ev$ s.t. $\cov, \ev \Vdash \phi$.
    Unpacking this, the resultlet at every $\tau_\vi \in \Sigma$ gives
    \begin{itemize}
        \item Morphism $\rho_\vi : \tau_\vi \ato \tau_\vi'$, finitary computation $\phi_\vi : \monFLSFin\vX\anyL{\tau_\vi'}$,
              and heap $\mvarM_\vi : \theapF{\tau_\vi'}$. 
        \item Every $\phi_\vi$ further gives a $\Sigma_\vi \in \famC\tinstC$, complete cover $\cov_\vi$ over $\pair{\anyL{\tau_\vi'}}{\Sigma_\vi}$
              and $\vX$-valued evaluation $\ev_\vi$ over $\Sigma_\vi$.
    \end{itemize}
    Firstly, we define its family of templates $\Sigma^{\mu\phi}$ by unifying every $\mvarM_\vi :\, ?^{\anyL{\tau_\vi'}} \to \tau_\vi'$ 
    with $\cov_\vi$, which is given by the following pushout in $\famC\tinstC$. 
    We then select only templates
    $\tau_\vi^\vk \in \Prod{\vi \in \indF\Sigma} (\tau_\vi' \liplus{?_{\anyL{\tau_\vi'}}} \Sigma_\vi)$ that are compatible with allocation $\me_\vi$ ($\rho_\vi$ with the allocation data):
    when the right-most square is a local independent general unifier. 
\begin{equation*}
    \begin{tikzcd}[sep=small]
        {\Prod{\vi \in \indF\Sigma} ?_{\anyL{\tau_\vi'}}}
            \ar[rr, "\Prod{\vi \in \indF\Sigma} \cov_\vi"] \ar[dd, "\Prod{\vi \in \indF\Sigma} \mvarM_\vi"'] 
            && {\Prod{\vi \in \indF\Sigma} \Sigma_\vi}
                \ar[dd, "\Prod{\vi \in \indF\Sigma} \ext^{\cov_\vi}\mvarM_\vi"] 
        \\\\
        {\Prod{\vi \in \indF{\Sigma}} \tau_\vi'}
            \ar[rr, "\Prod{\vi \in \indF\Sigma}\ext^{\mvarM_\vi}\cov_\vi"' yshift=-5pt]
            && {\Prod{\vi \in \indF\Sigma} (\tau_\vi' \liplus{?_{\anyL{\tau_\vi'}}} \Sigma_\vi)} \posquare{uull}
    \end{tikzcd}
    \begin{tikzcd}[sep=small]
        & {?_{\anyL{\tau_\vi'}}} \ar[dl, twoheadrightarrow, "\mvarM_\vi"'] \ar[dr, twoheadrightarrow, "\proj_\vj \comp \cov_\vi"]
    \\
    {\tau_\vi'} \ar[dr, twoheadrightarrow, "\me^\vk_1 \eqdef \proj_\vk \comp \ext^{\mvarM_\vi}\cov_\vi"']
            && {\tau_\vj} \ar[dl, twoheadrightarrow, "\proj_\vk \comp \ext^{\cov_\vi}\mvarM_\vi \defeq \me^\vk_2"]
    \\
        & {\tau_\vi^\vk}
    \end{tikzcd}
    \begin{tikzcd}[sep=small]
        & {\tau_\vi} 
            \ar[dl, twoheadrightarrow, end anchor=north east, shorten >=-2pt, "\con{\me^\vk_1 \comp \me_\vi}"'] \ar[dr, -alloctip, "\me_\vi"]
    \\
    {\tau_{\con{\me^\vk_1 \comp \me_\vi}}}
        \ar[dr, -rmonotip, "\uncon{\me^\vk_1 \comp \me_\vi}"']
            && {\tau_\vi'}
                \ar[dl, twoheadrightarrow, "\me^\vk_1"]
    \\
        & {\tau^\vk_\vi}
    \end{tikzcd}
\end{equation*}

$\Sigma^{\mu\phi}$ selects the $\tau_{\con{\me^\vk_1 \comp \me_\vi}}$ from every such compatible $\tau_\vi^\vk$
and morphism $\Sigma \xto{\embed} \Sigma^{\mu\phi}$ given by the $\con{\me^\vk_1 \comp \me_\vi}$ is complete.
This defines the cover $\cov^{\mu\phi} = \embed \comp \cov$ since connected morphisms are closed under composition.
Lastly, the evaluation $\ev^{\mu\phi}$ is given at every $\tau_{\con{\me^\vk_1 \comp \me_\vi}}$ as
\[ \proj_{\tau_{\con{\me^\vk_1 \comp \me_\vi}}} \ev^{\mu\phi} \eqdef 
        \hideRes\relax{\uncon{\me^\vk_1 \comp \me_\vi}} (\typeRes(\vX)(\me^\vk_2)(\proj_\vj \ev_\vi)) \]
Use the hiding algebras of $\hideMon$ (\cite{kammar2017monad}) to prove realisation.
\end{proofsketch}

\begin{theorem}[Finitarity]\label{thm:fin-fgls-monad-finitariness}
    Monad $\monFLSFin$ is finitary.
\end{theorem}
\begin{proofsketch}
    We need to show that $\monFLSFin$ preserves filtered colimits. 
    Let $\vF : \mcatI \to \pworldC$ be a filtered colimit and we need to show the induced cocone morphism $\Colimit(\monFLSFin \comp \vF) \dashto \monFLSFin(\Colimit\vF)$ is iso.
    We give the construction of the inverse.
    By finitarity of $\monFLSFin$, we have that a finitary $\phi : \monFLSFin(\Colimit\vF)(\vw)$ has a finite family of objects in $\mcatI$, $\vR \in \Prod{\vi \in \indF{\Sigma_\phi}} \objC\mcatI$.
    By filteredness, we can take the upper bound of these objects, denote it by $\vk$, with morphisms $\vR\vi \xsto{\vf_\vi} \vk$ for every $\vi \in \indF{\Sigma_\phi}$.
    We then extend every resultlet $\proj_\tau \ev_\phi \eqdef \inj_{\rho : \tau \to \tau'} \pair{[\vx]}\mvarM$ to $\inj_\rho\pair{\vF_{\anyL{\tau'}}(\vf_\vi)(\vx)}\mvarM$.
    We can show that this map (i) is invariant under which cover, evaluation pair that realises $\phi$, and (ii) is the inverse morphism from the colimiting cocone.
\end{proofsketch}


\section{Finitary Semantics for FGLS}\label{sec:fin-sem-fgls}

In this section we give semantics to $\lambdaref$ from \cref{sec:pl-lamref} using the finitary FGLS monad.
We interpret types as functors $\sem\vA : \worldC \to \setC$ that carry a covariant structure.
We show that generic effects $\genRead$, $\genWrite$ and $\genAlloc$ are finitary computations 
allowing us to interpret $\lambdaref$ terms. 
Lastly, we show a correspondence between the semantics for $\lambdaref$ given via the two monads. 

\subsection{Semantics of Types}

\Cref{fig:sem-types} gives the interpretation of $\lambdaref$ types as covariant functors $\sem\vA : \worldC \to \setC$:
given a $\sigma : \vw \to \vw'$, we can extend a $\va \in \sem\vA\vw$ to a $\sem\vA(\sigma)(\va) \in \sem\vA(\vw')$.
For $\tyZero,\tyUnit,\tyProd,\tySum,\tyFun$ the interpretation is standard using the bi-cartesian closed structure of $\pworldC$ and monad $\monFLSFin$ \cite{moggi89computational,moggi1990abstract}.
For reference types $\sem{\tyRef\vC}$, we use the representable functor $\homF\worldC{\vl : \vC}\dash$.
Thus a $\sigma \in \sem{\tyRef\vC}(\vw)$ is an injection $\vl : \vC \xto\sigma \vw$ giving us a single location in $\vw$.
The covariant structure of $\sem{\tyRef\vC}$ says that location $\vl$ remains available for every extension of $\vw$. 

\begin{figure}
\begin{center}
    \[\begin{array}{rclrcl}
        \sem\tyZero     & = & \typeZero                      & \sem{\vA \tyProd \vB} & = & \sem\vA \typeProd \sem\vB \\
        \sem\tyUnit     & = & \typeUnit                      & \sem{\vA \tySum \vB}  & = & \sem\vA \typeSum \sem\vB \\
        \sem{\tyRef\vC} & = & \homF\worldC{\vl : \vC}\dash   & \sem{\vA \tyFun \vB}  & = & \Exp{(\monFLSFin\sem\vB)}{\sem\vA}
    \end{array}\]
\end{center}
\caption{Interpretation of $\lambdaref$ types in $\pworldC$.}
\Description{Interpretation of $\lambdaref$ types in $\pworldC$.}
\label{fig:sem-types}
\end{figure}

\subsection{State Manipulating Generic Effects}\label{sec:state-generic-effects}

\paragraph{Generic Effect Read}
For every sort $\vC \in \typeSort$, monad $\monFLS$ supports generic effect $\genRead_\vC : \sem{\tyRef\vC} \to \monFLS{\sem{\ctype\vC}}$
which returns the contents of location $\vl$:
\[ \genRead_\vC(\vl)(\vw \xto\sigma \vw')(\mvarM : \heapF\vw') = \eqclass{\vw' \xsto1 \vw'}\pair{\proj_{\sigma(\vl)}\mvarM}\mvarM \]

$\genRead_\vC(\vl)$ is finitary since it only inspects a single location on the heap and full ground types have finite inhabitants.
Formally we use the fact that sorts $\vC$ with full ground interpretation can be enumerated~(see \cref{def:sort-enumeration}).
Denote its enumeration as $?^{\vl : \vC} \xto\ve \Sigma_\vC$.
By yoneda we have $\genRead_\vC \in \monFLS{\sem{\ctype\vC}}(\vl : \vC)$, and let its complete cover be $\ve$.
The evaluation at every $\tau_\vi \in \Sigma_\vC$ is resultlet $\inj{\tau_\vi \xsto1 \tau_\vi}\pair{\proj_{\anyL{\ve_\vi}(\vl)}\mvarM_{\ve_\vi}}{\mvarM_{\ve_\vi}}$.

\paragraph{Generic Effect Write}
Morphism $\genWrite_\vC : \sem{\tyRef\vC} \times \sem{\ctype\vC} \to \monFLS\typeUnit$ is given as
\[ \genWrite_\vC\pair\vl\vd(\vw \xto\sigma \vw')(\mvarM : \heapF\vw') = \eqclass{\vw' \xsto1 \vw'}\pair\Unit{\mvarM[\sigma(\vl) \mapsto \sem{\ctype\vC}(\sigma)(\vd)]} \]
$\genWrite_\vC\pair\vl\vd \in \monFLS\typeUnit\vw$ lifts through $\monFLSFin$ monad as follows.
Notice that $\vl$ is a morphism $(\vl : \vC) \xto\vl \vw)$ and let $?^{\vl : \vC} \xto\ve \Sigma_\vC$ be an enumeration of $\vC$ as before.
We first rebase cover $\ve$ (\cref{def:cov-rebase}) via $(\vl : \vC) \xto\vl \vw$, which induces a complete cover $\ve^\vl$ over $\pair\vw{\Sigma_\vC^\vl}$.
Define the evaluation at every $\tau_\vi \in \Sigma_\vC^\vl$ as resultlet 
\[ \inj{\tau_\vi \xsto1 \tau_\vi}\pair\Unit{\mvarM_{\me_\vi}[\anyL{\ve^\vl_\vi}(\vl) \mapsto \sem{\ctype\vC}[\ve^\vl_\vi](\vd)]} \]

\paragraph{Generic Effect Alloc}

Recursive mutual allocation of locations is given by generic allocation $\genAlloc$.   
For every world $\vw_0$ to be allocated, we have 
$\genAlloc_{\vw_0} : \Prod{\vl : \vC \in \vw_0}\sem{\ctype\vC}\extF\dash{\vw_0} \to \monFLS\homF\worldC{\vw_0}\dash$.
The extension functor $\extF\dash{\vw_0}$ adds $\vw_0$-many locations, and so at a world $\vw$, 
we have $\vw_0$-many values $\vp \in \Prod{\vl : \vC \in \vw_0}\sem{\ctype\vC}(\vw \oplus \vw_0)$. 
Notice that $\vp$ is equivalent to a morphism $\vw \xto\me \vw \oplus \vw_0$ in $\instC$. 
Therefore, $\genAlloc_{\vw_0}$ is given as
\[ \genAlloc_{\vw_0}(\vp)(\vw \xto\sigma \vw')(\mvarM : \heapF\vw') = \eqclass{\ext^\sigma\me}\pair{\ext^\me\sigma \comp \incl_2}{\heapF(\ext^\sigma\me)(\mvarM)} \]
using the local independent coproduct
\begin{tikzcd}[sep=small]
    {\vw} \ar[r, "\me"] \ar[d, "\sigma" xshift=1pt]
        & {\vw \oplus \vw_0} \ar[d, "\ext^\me\sigma" xshift=2pt] 
    \\
    {\vw'} \ar[r, "\ext^\sigma\me" yshift=2pt]
        & {\vw' \liplus\vw (\vw \oplus \vw_0)}
\end{tikzcd} and notice that $\ext^\sigma\me$ can be equipped with the data of $\me$ (see \cite{kammar2017monad}).

We show that $\genAlloc_{\vw_0}(\vp) \in \monFLS{\homF\worldC{\vw_0}\dash}(\vw)$ is finitary by selecting the identity cover $?^\vw \xTo1 ?^\vw$.
For the resultlet, let $\incl_2^{\vw \oplus \vw_0} : \vw_0 \to \vw \oplus \vw_0$ be the template with $\vw \oplus \vw_0$ locations and $\img{\incl_2}$ concrete ones. 
We have an allocation morphism $\rho :\, ?^{\vw} \to \incl_2^{\vw \oplus \vw_0}$ in $\tinstC$.
Define the resultlet as $\inj_{\rho}\pair{\incl_2}{\mvarM_\me}$.

\subsection{Semantics of Terms}

Typing judgments are interpreted as natural transformations in $\pworldC$:
\[ \semV{\vtyped\Gamma\vw\vM\vA} : \homF\worldC\vw\dash \times \sem\Gamma \to \sem\vA \quad\text{ and }\quad
   \sem{\typed\Gamma\vw\vM\vA}   : \homF\worldC\vw\dash \times \sem\Gamma \to \monFLSFin\sem\vA \]
Contexts are interpreted component-wise $\sem\Gamma \eqdef \Prod{\vx : \vA \in \Gamma}\sem\vA$, and
world $\vw$ in judgment $\typed\Gamma\vw\vM\vA$ is modelled using the representable functor $\homF\worldC\vw\dash$:
it defines $\sem{\typed\Gamma\vw\vM\vA}$ only at worlds $\vw'$ that have an injection $\vw \injto \vw'$, \ie worlds bigger than $\vw$.

The interpretation of values (\Cref{fig:sem-values} is mostly standard.
It expects a location environment $\sigma$ that associates location literals $\vl$ with the current world,
and a variable environment $\gamma$ to interpret variables in context.
\begin{figure}
    \[\begin{array}{r@{\,}l@{\quad}r@{\,}l}
        \multicolumn{2}{r}{\semV{\vtyped\Gamma\vw\vV\vA} :} & \multicolumn{2}{@{}l}{\homF\worldC\vw\dash \times \sem\Gamma \to \sem\vA} \\
    \semV{\vtyped{\vx : \vA, \Gamma}\vw\vx\vA}\pair\sigma\gamma& = \gamma(\vx) &
            \semV{\vtyped\Gamma\vw\Unit\tyUnit}\pair\sigma\gamma & = \star \\
    \semV{\vtyped\Gamma\vw{\Lam\vx\vM}{\vA \tyFun \vB}}\pair\sigma\gamma& = \curry\sem\vM\pair\sigma\gamma &
\semV{\vtyped\Gamma\vw{\Pair\vV\vW}{\vA \tyProd \vB}}\pair\sigma\gamma & = \pair{\semV\vV}{\semV\vW}\pair\sigma\gamma \\
    \semV{\vtyped\Gamma\vw\vl{\tyRef\vC}}\pair\sigma\gamma & = \sigma(\vl) &
        \semV{\vtyped\Gamma\vw{\Inj_\vi\vV}{\vA \tySum \vB}}\pair\sigma\gamma & = \inj_\vi(\semV\vV\pair\sigma\gamma) \\
\end{array}\]
\caption{Interpretation of $\lambdaref$ values in $\pworldC$}
\Description{Interpretation of $\lambdaref$ values in $\pworldC$}
\label{fig:sem-values}
\end{figure}

For the denotation of expressions (\cref{fig:sem-expr}), first denote the strength as $\str_{\vX,\vY} : \vX \times \monFLSFin\vY \to \monFLSFin(\vX \times \vY)$, 
and the kleisli extension of a morphism $\vX \xsto\vN \monFLSFin\vY$ as 
$\bind\vN \eqdef \mu \comp \monFLSFin\vN : \monFLSFin\vX \to \monFLSFin\vY$.
Further, define the left-to-right double strength $\dstr$ as (where $\swap_{\vX,\vY} \eqdef \pair{\proj_2}{\proj_1} : \vX \times \vY \to \vY \times \vX$):
\[
    \dstr_{\vX,\vY} \bind(\str \comp \swap) \comp \str \comp \swap : \monFLSFin\vX \times \monFLSFin\vY \to \monFLSFin(\vX \times \vY)
\]
Lastly, let $?_\vX : \typeZero \to \vX$ be the unique morphism from the initial object, and
let $\distr_{\vX,\vY,\vZ} : \vX \times (\vY + \vZ) \to (\vX \times \vY) + (\vX \times \vZ)$ be the distribution of products over sums given by the bi-ccc structure.

Semantically, a value $\semV{\vtyped\Gamma\vw\vW\vA}$ is an expression $\sem{\typed\Gamma\vw\vV\vA}$ by postcomposing with unit $\eta$.
Constructs for applications, products and sums are interpreted in a standard way.
For equality testing on references we use $\eqOp : \sem{\tyRef\vC} \times \sem{\tyRef\vC} \to \typeBool$ defined as $\pair{\sigma_1}{\sigma_2} \mapsto \sigma_1(\vl) == \sigma_2(\vl)$.
For reading, writing and allocation we use the generic effects given at \cref{sec:state-generic-effects}.
For allocation, we first lift environment $\gamma \in \sem\Gamma(\vw')$ to $\sem\Gamma(\incl_1)(\gamma) : \sem\Gamma(\vw' \oplus \vw_0)$ and extend it to 
$\gamma' : \sem{\Gamma,\vx_1 : \tyRef\vC, \hdots, \vx_\vn : \tyRef{\vC_\vn}}(\vw' \oplus \vw_0)$.
We then apply $\genAlloc$ with values $\sem{\vV_\vi}\pair{\sigma \oplus \vw_0}{\gamma'} : \sem{\ctype\vC}(\vw' \oplus \vw_0)$ where $\sigma \oplus \vw_0 : \vw \oplus \vw_0 \to \vw' \oplus \vw_0$.

\begin{figure}
\begin{align*}
    \sem{\typed\Gamma\vw\vM\vA}                                    & :\, \homF\worldC\vw\dash \times \sem\Gamma \to \monFLSFin\sem\vA \\
    \sem{\typed\Gamma\vw\vV\vA}                                    & =\, \eta_{\sem{\vA}} \comp \semV\vV \\
    \sem{\typed\Gamma\vw{\vM\App\vN}\vA}                           & =\, \bind\eval \comp \dstr_{\sem{\vA \to \vB}, \sem\vA} \comp \pair{\sem\vM}{\sem\vN} \\
    \sem{\typed\Gamma\vw{\Split\vM{\vx_1}{\vx_2}\vN}\vB}           & =\, \bind{\sem\vN} \comp \str \comp \pair{1}{\sem\vM} \\
    \sem{\typed\Gamma\vw{\Pair\vM\vN}{\vA \tyProd \vB}}            & =\, \dstr \comp \pair{\sem\vM}{\sem\vN} \\
    \sem{\typed\Gamma\vw{\Inj_\vi \vM}{\vA \tySum \vB}}            & =\, \monFLSFin(\incl_\vi) \comp \sem\vM \\
    \sem{\typed\Gamma\vw{\Absurd\vM}\vB}                           & =\, \monFLSFin(?_{\sem\vB}) \comp \sem\vM \\
    \sem{\typed\Gamma\vw{\Case\vM{\vx_1}{\vN_1}{\vx_2}{\vN_2}}\vB} & =\, \bind\copair{\sem{\vN_1}}{\sem{\vN_2}} \comp \monFLSFin\distr \comp \str \comp \pair{1}{\sem\vM} \\
    \sem{\typed\Gamma\vw{\vM\Eq\vN}\tyBool}                        & =\, \monFLSFin\eqOp \comp \dstr \comp \pair{\sem\vM}{\sem\vN} \\
    \sem{\typed\Gamma\vw{\Read\vM}{\ctype\vC}}                     & =\, \bind\genRead \comp \sem\vM \\
    \sem{\typed\Gamma\vw{\vM\Write\vN}\tyUnit}                     & =\, \bind\genWrite \comp \dstr \comp \pair{\sem\vM}{\sem\vN} \\
    \sem{\typed\Gamma\vw{\Letref{\vx_1 \Write \vV_1, \hdots, \vx_\vn \Write \vV_\vn}\vM}\vB}\pair\sigma\gamma & =
    \str(\pair\sigma\gamma, \genAlloc\langle\sem{\vV_\vi}\pair{\sigma\oplus{\vw_0}}{\gamma'}\rangle_{\vi = 1}^\vn) \bind{\sem\vM} \\
                & \quad\qquad \mwhere \gamma' \eqdef  \sem\Gamma(\incl_1)(\gamma)[\vx_\vi \mapsto \incl_2(\vl_\vi)]_{\vi = 1}^\vn
\end{align*}
\caption{Interpretation of $\lambdaref$ expressions in $\pworldC$}
\Description{Interpretation of $\lambdaref$ expressions in $\pworldC$}
\label{fig:sem-expr}
\end{figure}

\subsection{Metatheory}

Let $\semM\dash\monFLS$ denote the interpretation of $\lambdaref$ terms using the full monad $\monFLS$, given identically to the one above.
At full-ground types $\vD \in \typeFG$, the denotations are identical $\semM\vD\monFLS = \semM\vD\monFLSFin$, 
whereas at higher types they differ since $\semM{\vA \to \vB}\monFLS \neq \semM{\vA \to \vB}\monFLSFin$.

Recall from \cref{sec:fin-monad} that $\monFLSFin$ is related with $\monFLS$ using a strong monad morphism $\iota$. 
We can therefor establish a relationship between denotations of first-order terms given by the two monads~\cite{katsumata2013relating}.

\begin{theorem}\label{thm:fo-sem-agree}
    For every first-order term $\vx_1: \vD_1, \hdots, \vx_\vn : \vD_\vn \vdash_\vw \vM : \vD$
    we have $\iota_{\sem\vD} \comp \semM\vM\monFLSFin = \semM\vM\monFLS$.
\end{theorem}


Adequacy from the $\semM\dash\monFLS$ semantics extends to $\semM\dash\monFLSFin$ using \Cref{thm:fo-sem-agree}.
Soundness of the semantics w.r.t. the operational semantics and compositionality is proven in the standard way by inducting over the derivations.


\section{Related Work}


\citet{moggi1990abstract} gave a unified categorical semantics for modelling various computational effects and specifically 
gave a possible worlds monad on a functor category for dynamic allocation of ground storage.
Expanding on Moggi's work, \citet{plotkin2002notions} considered monads induced by suitable equations over the effectful operations, giving an equational account of computational effects. 
They further define a richer (ground) local state monad that supported garbage collection capabilities.
Our work is based on the FGLS monad of \citet{kammar2017monad} that follows along this direction by giving a monad for \emph{full ground} local state.
\citet{polzer2020local} complement the semantics of \cite{kammar2017monad} by developing a higher-order logic over the FGLS monad based on bi-hyperdoctorines.

Multiple relational models for local state exist in the literature that rely on auxiliary logical relations built on top of the semantics
~\cite{benton2005relational,bohr2006relational,benton2007relational,benton2014abstract,birkedal2011step,birkedal2010realisability}.
These vary from utilising a continuation monad over nominal-cpos~\cite{benton2005relational,bohr2006relational},
incorporating region-based effect typing in their reasoning~\cite{benton2007relational,benton2014abstract},
and consider kripke logical relations over recursively defined set of worlds using metric spaces based semantics~\cite{birkedal2010realisability}.

In any case, the relative simplicity of the underlying denotational semantics means that they fail to validate desired equivalences involving unobservable allocations. 
Instead it allows them to model complex features such as general references, recursive types and parametric polymorphism.
They employ logical relations to refine their model allowing them to reason about properties such as observational equivalence.
In contrast, the possible worlds semantics of \cite{plotkin2002notions,kammar2017monad} builds garbage collection into the model, 
complicating the model but avoid requiring auxiliary logical relations.

Syntactic step-indexed logical relations built over the operational semantics have also been successful in modelling multiple non-trivial programming language constructs in conjunction such as recursive, abstract, polymorphic types, general references and control effects~\cite{ahmed2004semantics,ahmed2009state,dreyer2012impact,birkedal2011step}.
By marrying the reasoning with separation logic~\cite{jung2018iris,timany2024logical} they are able to prove deep properties of higher-order concurrent programs that make use of general references.
We instead focus on denotational accounts of local state that abstract away from lower-level operational considerations.

An alternative class of models is based on game semantics, which differs from the aforementioned models as the semantics of heaps is implicit~\cite{abramsky1998fully,laird2008game,murawski2012algorithmic}.
In this line of work, execution of a program is modelled as a dialogue between the player representing the program and the opponent being the environment.
Game semantics is suitable for modelling full ground and general references and produces models that are fully abstract at higher-order types.
Nevertheless, these models are fundamentally different from the possible worlds semantics we are concerned with. 


\section{Conclusion and Future Work}

We extend the metatheory of the FGLS monad introduced in \cite{kammar2017monad} by resolving an outstanding question on whether it is finitary.
We show a negative result and therefore introduce a finitary FGLS submonad. 
This submonad paves the way towards finitary equational axiomatisations of FGLS that are more amenable to reasoning.

There are several directions for future work. 
A useful property for a semantics is full abstraction at first order types~\cite{kammar2026equational,abramsky1998fully,milner1977fully,streicher2006domain}.
It states that for terms $\typed\Gamma\relax{\vM, \vN}\vA$, observational equivalence implies denotational equality: $\vM \obsEq \vN \implies \sem\vM = \sem\vN$, where $\Gamma, \vA$ are first order, \ie do not have any function types.
Together with adequacy, full abstraction at first order identifies observational and denotational equality of first order terms.
Such a property is used to show that the model is complete for compiler optimisations on imperative programs involving full ground references.
This is typically the best-possible result, as models given by sets with structure and structure preserving functions are not typically fully abstract at higher order.
We conjecture that the FGLS monad is fully abstract at first order, and that we can use techniques inspired by game semantics to prove that this is the case.
This result would then transfer to the finitary submonad, and we can exploit its finitary nature to develop decision procedures to decide whether two terms are observationally equivalent.

We also aim to generalise the equations for ground local state \cite{plotkin2002notions,staton2010completeness,staton2013instances} to the full-ground setting. 
This is a non-trivial task for two reasons: (i) one must identify an appropriate framework in which heap values can be treated as first-class entities, 
and (ii) one must account for the symmetries present within an allocated block of data. 
The goal is to obtain an equational axiomatisation giving rise to the finitary FGLS monad developed here.

Furthermore, we aim to increase the expressivity of the FGLS monad by reformulating it over $ \worldC $-indexed functors into cpos. 
This would allow us to support general recursion and would be a further step towards modelling higher-order store. 
As an additional benefit, such a domain-theoretic presentation of the FGLS monad may offer an alternative characterisation of the finitary submonad in terms of abstract domain constructions.

Finally, it would be worthwhile to investigate which programming-language constructs can express complete heap traversals without introducing a primitive operation that returns the size of the heap, such as $ \heapSize $ from \cref{sec:overview-non-finitariness}. 
We expect, however, that such an operation would have no substantial impact on the metatheory: 
the FGLS monad would still fail to be fully abstract at higher-order types.


\begin{acks}
Ohad Kammar was funded by a Royal Society University Research Fellowship, 
Sam Lindley was supported by a UKRI Future Leaders Fellowship “Effect Handler Oriented Programming” (MR/T043830/1 and MR/Z000351/1),
and Cristina Matache was partly supported by a Leverhulme Early Career Fellowship.
We thank the Paul B. Levy for useful discussions and suggestions.
\end{acks}

\bibliographystyle{ACM-Reference-Format}
\bibliography{bibliography.bib}

\clearpage
\pagebreak
\appendix

\section{Appendix}

\subsection[Full Specification of Programming Language]{Full Specification of $\lambdaref$}\label{sec:lambdaref}

\begin{figure}[H]
\begin{align*}
    \typeType \ni \vA,\vB & \inddef \gbox{\tyRef\vC} \mid     
                            \tyZero \mid \vA \tySum \vB \mid 
                            \tyUnit \mid \vA \tyProd \vB \mid
                            \vA \tyFun \vB                  
    \\
    \typeVal \ni \vV,\vW  \inddef &\, \vx \mid \Lam\vx\vM \mid       & \tag{$\lambda$-calculus} \\
                                  &\, \gbox\vl \mid                  & \tag{Location Literals} \\
                                  &\, \Unit \mid \Pair\vV\vW \mid    & \tag{Products} \\
                                  &\, \Inj_{\vi \in \{1,2\}} \vV     & \tag{Sums}
    \\
    \typeExpr \ni \vM,\vN  \inddef &\, \vV \mid \vM\App\vN \mid & \tag{Values and Application} \\
                                   &\, \Pair\vM\vN \mid \Split\vM{\vx_1}{\vx_2}\vN \mid & \tag{Products} \\
                                   &\, \Inj_{\vi \in \{1,2\}} \vM \mid \Absurd\vM \mid \Case\vM{\vx_1}{\vN_1}{\vx_2}{\vN_2} \mid & \tag{Sums} \\
                                   &\, \gbox{\Read\vM} \mid \gbox{\vM\Write\vN} \mid \gbox{\Letref{\vx_1 \Write \vV_1, \hdots, \vx_\vn \Write \vV_\vn}\vM} & \tag{References}
\end{align*}
\Description{Syntax of programming language $\lambdaref$}
\caption{Syntax of $\lambdaref$}
\end{figure}

\begin{figure}[H]
\begin{mathpar}
\inferrule[\nameTVar]{\vx : \vA \in \Gamma}{\vtyped\Gamma\vw\vx\vA}

\inferrule[\nameTLam]{\typed{\Gamma, \vx : \vA}\vw\vM\vB}{\vtyped\Gamma\vw{\Lam\vx\vM}{\vA \tyFun \vB}}

\gbox{
\inferrule[\nameTLoc]{\vl : \vC \in \vw}{\vtyped\Gamma\vw\vl{\tyRef\vC}}
}

\inferrule[\nameTUnit]{}{\vtyped\Gamma\vw\Unit{\tyUnit}}

\inferrule[\nameTVPair]{
    \vtyped\Gamma\vw\vV\vA \\\\
    \vtyped\Gamma\vw\vW\vB
}{\vtyped\Gamma\vw{\Pair\vV\vW}{\vA \tyProd \vB}}

\inferrule[\nameTVInj]{\vtyped\Gamma\vw\vV{\vA_\vi}}{\vtyped\Gamma\vw{\Inj_\vi\vV}{\vA_1 \tySum \vA_2}}

\inferrule[\nameTVal]{\vtyped\Gamma\vw\vx\vA}{\typed\Gamma\vw\vx\vA}

\inferrule[\nameTApp]{
    \vtyped\Gamma\vw\vM{\vA \tyFun \vB} \\\\
    \vtyped\Gamma\vw\vN\vA
}{\typed\Gamma\vw{\vM\App\vN}\vB}

\inferrule[\nameTEPair]{
    \typed\Gamma\vw\vM\vA \\\\
    \typed\Gamma\vw\vN\vB
}{\typed\Gamma\vw{\Pair\vM\vN}{\vA \tyProd \vB}}

\inferrule[\nameTSplit]{
    \typed\Gamma\vw\vM{\vA_1 \tyProd \vA_2} \\\\
    \typed{\Gamma, \vx_1 : \vA_1,\vx_2 : \vA_2}\vw\vN\vB
}{\typed\Gamma\vw{\Split\vM\vx\vy\vN}\vB}

\inferrule[\nameTEInj]{\typed\Gamma\vw\vM{\vA_\vi}}{\typed\Gamma\vw{\Inj_\vi\vM}{\vA_1 \tySum \vA_2}}

\inferrule[\nameTCase]{
    \typed\Gamma\vw\vM\tyZero
}{\typed\Gamma\vw{\Absurd\vM}\vA}

\inferrule[\nameTCase]{
    \typed\Gamma\vw\vM{\vA_1 \tySum \vA_2} \\\\
    (\typed{\Gamma, \vx_\vi : \vA_\vi}\vw{\vN_\vi}\vB)_{\vi = 1}^2
}{\typed\Gamma\vw{\Case\vM{\vx_1}{\vN_1}{\vx_2}{\vN_2}}\vB}

\gbox{
\inferrule[\nameTRead]{\typed\Gamma\vw\vM{\tyRef\vC}}{\typed\Gamma\vw{\Read\vM}{\ctype\vC}}
}

\gbox{
\inferrule[\nameTWrite]{
    \typed\Gamma\vw\vM{\tyRef\vC}
    \\\\
    \typed\Gamma\vw\vN{\ctype\vC}
}{\typed\Gamma\vw{\vM\Write\vN}\tyUnit}
}

\gbox{
\inferrule[\nameTAlloc]{
    (\typed{\Gamma, \vx_1 : \tyRef{\vC_1}, \hdots, \vx_\vn : \tyRef{\vC_\vn}}\vw{\vV_\vi}{\ctype\vC_\vi})_{i=1}^\vn
    \\\\
    \typed{\Gamma, \vx_1 : \tyRef{\vC_1}, \hdots, \vx_\vn : \tyRef{\vC_\vn}}\vw\vM\vA
}{\typed\Gamma\vw{\Letref{\vx_1 \Write \vV_1, \hdots, \vx_\vn \Write \vV_\vn}\vM}\vA}
}
\end{mathpar}
\Description{Typing rules of programming language $\lambdaref$}
\caption{Typing rules of $\lambdaref$}
\end{figure}

\begin{figure}[H]
\begin{mathpar}
    \inferrule[\nameRVal]{}{\bigstep\vV\mvarM\vV\mvarM}
    
    \inferrule[\nameRApp]{
        \bigstep\vM\mvarM{\Lam\vx{\vM'}}{\mvarM'}
        \and
        \bigstep\vN{\mvarM'}\vV{\mvarM''}
        \\\\
        \bigstep{\vM'\Subst\vV\vx}{\mvarM''}\vW{\mvarM'''}
    }{
        \bigstep{\vM\App\vN}\mvarM\vW{\mvarM'''}
    }

    \inferrule[\nameREPair]{
        \bigstep\vM\mvarM\vV{\mvarM'}
        \and
        \bigstep\vN{\mvarM'}\vW{\mvarM''}
    }{
        \bigstep{\Pair\vM\vN}\mvarM{\Pair\vV\vW}{\mvarM''}
    }

    \inferrule[\nameRSplit]{
        \bigstep\vM\mvarM{\Pair{\vV_1}{\vV_2}}{\mvarM'}
        \\\\
        \bigstep{\vN\SubstTwo{\vx_1}{\vV_1}{\vx_2}{\vV_2}}{\mvarM'}\vW{\mvarM''}
    }{
        \bigstep{\Split\vM{\vx_1}{\vx_2}\vN}\mvarM\vW{\mvarM''}
    }

    \inferrule[\nameREInj]{
        \bigstep\vM\mvarM\vV{\mvarM'}
    }{
        \bigstep{\Inj_\vi\vM}\mvarM{\Inj_\vi\vV}{\mvarM'}
    }

    \inferrule[\nameRCase]{
        \bigstep\vM\mvarM{\Inj_\vi\vV}{\mvarM'}
        \\\\
        \bigstep{{\vM_\vi}\Subst{\vx_\vi}\vV}{\mvarM'}\vW{\mvarM''}
    }{
        \bigstep{\Case\vM{\vx_1}{\vN_1}{\vx_2}{\vN_2}}\mvarM\vW{\mvarM''}
    }
\gbox{
    \inferrule[\nameRRead]{
        \bigstep\vM\mvarM\vl{\mvarM'}
        \and
        \mvarM'(\vl) = \vV
    }{
        \bigstep{\Read\vM}\mvarM\vV{\mvarM'}
    }
}

\gbox{
    \inferrule[\nameRWrite]{
        \bigstep\vM\mvarM\vl{\mvarM'}
        \and
        \bigstep\vN{\mvarM'}\vW{\mvarM''}
    }{
        \bigstep{\vM \Write \vN}\mvarM\Unit{\mvarM''[\vl \mapsto \vW]}
    }
}

\gbox{
    \inferrule[\nameREq]{
        \bigstep\vM\mvarH{\vl_1}{\mvarH'}
        \\\\
        \bigstep\vN{\mvarH'}{\vl_2}{\mvarH''}
    }{
        \bigstep{\vM \Eq \vN}\mvarH{\#(\vl_1 == \vl_2)}{\mvarH''}
    }
}

\gbox{
    \inferrule[\nameRAlloc]{
        \vl_1, \hdots, \vl_\vn \not\in \dom\mvarM
        \\\\
        \theta = [\vx_\vi \mapsto \vl_\vi]_{\vi = 1}^\vn
        \and
        \mvarM' = \mvarM[\vl_\vi \mapsto \vV_\vi[\theta]]_{\vi = 1}^\vn
        \and
        \bigstep{\vM[\theta]}{\mvarM'}\vW{\mvarM''}
    }{
        \bigstep{\Letref{\vx_1 \Write \vV_1, \hdots, \vx_\vn \Write \vV_\vn}\vM}\mvarM\vW{\mvarM''}
    }
}
\end{mathpar}
\Description{Big-step style reduction rules of programming language $\lambdaref$}
\caption{Reduction rules of $\lambdaref$ in big-step}
\end{figure}

\subsection{Non-Finitariness of FGLS Monad}

\begin{lemma}\label{lem:non-finitariness-proof}
    Monad $\monFLS$ does not preserve filtered colimits.
\end{lemma}
\begin{proof}
Recall the local state function $\heapSize$ from \cref{sec:overview-non-finitariness}, and for simplicity, assume we have only one sort $\sortNat$.
It is straightforward to extend $\heapSize$ to any interpretation function defined over the full ground types of \cref{sec:pl-lamref}.

$\heapSize$ takes a location $\vl : \sortNat$, traverses the heap starting from $\vl$ and returns the number of distinct locations it has found.
Let $\Delta_\typeNat : \worldC \to \setC$ be the constant functor $\Delta_\typeNat\vw = \typeNat$ in
\[\begin{array}{ll}
    &\heapSize : \End{\vl : \sortNat \injto \vw \in \sliceC{\vl : \sortNat}\vU} \heapF\vw \to \hideMon(\Delta_\typeNat \times \heapF)(\vw) \\
    &\heapSize((\vl : \sortNat) \xsto\rho \vw)(\vh : \heapF\vw) = \eqclass{\vw \leq \vw}\pair\vn\vh 
    \\
        \mwhere & 
    \\
                &  \go = \Lam{\pair{\vl \in \vw}{\vS \subseteq \vw}} 
                                                  \CaseV{\proj_\vl\vh}
                                                            {\InjL\Unit}{1}
                                                            {\InjR\vb}{\begin{cases}1 & \mif \vb \not\in \vS \\
                                                                                    1 + \go(\vb, \vS \setDiff \{\vb\})  & \motherwise\end{cases}} 
    \\
                &  \vn = \go\pair{\rho(\vl)}\vw
\end{array}\]
Notice that $\mathit{go}$ is well-founded since $\vS$ strictly reduces in size at every recursive call.
It also satisfies the coherence condition of the End as it only traverses the reachable part of the heap w.r.t. public location $\vl$. 

To show that $\heapSize$ is a counterexample to $\monFLS$ being finitary we need a filtered colimit that is not preserved by $\monFLS$.
Consider the (filtered) posetal category $\omega$ of natural numbers with the usual ordering, and let $\Delta_\dash : \Colimit{(\omega \to \funC\worldC\setC)}$ be the functor that maps each $\vn$ to its down-set, $\Delta_\vn = \downset\vn$.
$\monFLS$ does not preserve the (filtered) $\Colimit{\Delta_\dash} : \worldC \to \setC$ because the mapping given by the appropriate inclusion
\[ \disunion{\vn \in \typeNat}\monFLS(\Delta_{\downset\vn})(\vw) \to \monFLS(\Delta_\typeNat)(\vw) \]
is not surjective: there is no $\vn$ that defines $\heapSize$
\end{proof}

\end{document}